\pdfoutput=1
\documentclass[prologue,dvipsnames,acmsmall,screen]{acmart}
\newtoggle{appendix}
\toggletrue{appendix} 

\setcopyright{rightsretained}
\acmPrice{}
\acmDOI{10.1145/3591222}
\acmYear{2023}
\copyrightyear{2023}
\acmSubmissionID{pldi23main-p26-p}
\acmJournal{PACMPL}
\acmVolume{7}
\acmNumber{PLDI}
\acmArticle{108}
\acmMonth{6}
\received{2022-11-10}
\received[accepted]{2023-03-31}

\usepackage{algorithm}
\usepackage[noend]{algpseudocode} 
\usepackage{xspace} 
\usepackage{graphicx} 
\usepackage{float} 
\usepackage{subcaption} 
\usepackage[inline]{enumitem} 
\usepackage{siunitx} 
\usepackage{xcolor}
\usepackage{wrapfig} 
\usepackage{local} 
\usepackage{tikz}
\usepackage{tikz-network}
\usepackage{pgfplots} 
\usepackage{pgfplotstable} 

\pgfplotsset{compat=1.17}
\usepgfplotslibrary{groupplots} 
\usepgfplotslibrary{units} 
\SetVertexStyle[TextFont=\normalsize,FillColor=nodepurple]
\tikzstyle{timeout line} = [mark=none,black,dashed,samples=2,domain=20:2000]
\tikzstyle{timeout node} = [pos=0.2,yshift=-5pt,font=\it]

\pgfplotscreateplotcyclelist{cute}{
    {blue,mark=pentagon},
    {orange,mark=triangle},
    {green,mark=square},
    {red,mark=diamond},
}

\pgfplotstableread[search path={./results/}]{r-2022-10-20.dat}\rtbl
\pgfplotstableread[search path={./results/}]{ar-2022-10-20.dat}\artbl
\pgfplotstableread[search path={./results/}]{lw-2022-11-07.dat}\ltbl
\pgfplotstableread[search path={./results/}]{alw-2022-11-07.dat}\altbl
\pgfplotstableread[search path={./results/}]{v-2022-10-20.dat}\vtbl
\pgfplotstableread[search path={./results/}]{av-2022-10-20.dat}\avtbl
\pgfplotstableread[search path={./results/}]{h-2022-10-20.dat}\htbl
\pgfplotstableread[search path={./results/}]{ah-2022-10-20.dat}\ahtbl

\renewcommand{\new}[1]{#1}

\begin{document}

\title{Modular Control Plane Verification via Temporal Invariants}

\author{Timothy Alberdingk Thijm}
\orcid{0000-0003-1758-5917}
\affiliation{%
    \institution{Princeton University}
    \city{Princeton}
    \state{NJ}
    \country{United States}
}
\email{tthijm@cs.princeton.edu}
\author{Ryan Beckett}
\orcid{0000-0001-7844-2026}
\affiliation{%
    \institution{Microsoft Research}
    \city{Redmond}
    \state{WA}
    \country{United States}
}
\email{ryan.beckett@microsoft.com}
\author{Aarti Gupta}
\orcid{0000-0001-6676-9400}
\affiliation{%
    \institution{Princeton University}
    \city{Princeton}
    \state{NJ}
    \country{United States}
}
\email{aartig@cs.princeton.edu}
\author{David Walker}
\orcid{0000-0003-3681-149X}
\affiliation{%
    \institution{Princeton University}
    \city{Princeton}
    \state{NJ}
    \country{United States}
}
\email{dpw@cs.princeton.edu}

\ccsdesc[500]{Networks~Protocol testing and verification}
\ccsdesc[500]{Networks~Formal specifications}
\ccsdesc[500]{Theory of computation~Verification by model checking}
\ccsdesc[500]{Theory of computation~Automated reasoning}

\keywords{formal network verification, compositional reasoning, modular verification}

\begin{abstract}
  Monolithic control plane verification cannot scale to
  hyperscale network architectures with tens of thousands of nodes,
  heterogeneous network policies and thousands of network changes a day.
  Instead, \emph{modular verification} offers improved scalability,
  reasoning over diverse behaviors, and robustness following policy updates.
  We introduce \sysname{}, a new modular control plane verification system.
  While one class of verifiers, starting with Minesweeper, 
  were based on analysis of stable paths, we show that such models, when
  deployed naïvely for modular verification, are unsound.
  To rectify the situation, we adopt a
  routing model based around a logical notion of time and
  develop a sound, expressive, and scalable verification engine.

  Our system requires that a user specifies interfaces between module components.
  We develop methods for defining these interfaces using predicates inspired by temporal logic,
  and show how to use those interfaces to verify a range of network-wide properties
  such as reachability or access control.
  Verifying a prefix-filtering policy using a non-modular verification engine times
  out on an 80-node fattree network after 2 hours.
  However, \sysname{} verifies a 2,000-node fattree in 2.37 minutes on a 96-core virtual machine.
  Modular verification of individual routers is embarrassingly parallel and completes in seconds,
  which allows verification to scale beyond non-modular engines,
  while still allowing the full power of SMT-based symbolic reasoning.
\end{abstract}

\maketitle
\section{Introduction}
\label{sec:intro}

Major cloud providers are seeing sustained financial growth in response
to mounting demand for reliable networking~\cite{cloud-market-growth}.
This demand suggests a commensurate network \emph{infrastructure growth} will take place to
accommodate more and more users.
These networks can already have hundreds of data centers, each with hundreds of thousands of devices
running thousands of heterogeneous policies,
and receiving thousands of updates a day~\cite{secguru-slides}.
Network operators program this infrastructure using distributed routing protocols, where each router
in a network may run thousands of lines of configuration code.
Despite operators' care,
routine configuration updates have inadvertently rendered routers unreachable~\cite{cloudflare}
or violated isolation requirements that prevent flooding~\cite{rogers-register}.

To prevent costly errors, operators can use \emph{control plane verification} to
analyze their networks~\cite{arc,era,minesweeper,fastplane,bonsai,tiramisu,shapeshifter,plankton,hoyan,bagpipe}.
Until recently, research has focused on \emph{monolithic} verification of the entire network at once,
which is infeasible for large cloud provider networks.
Such networks demand \emph{modular} techniques that divide the network into components to verify in isolation.
This approach has proven successful for software verification~\cite{henzinger1998you,flanagan2003thread,grumberg1994model,alur1999reactive,giannakopoulou2018compositional}
and network data plane verification~\cite{secguru}.
We annotate the \emph{interfaces} between network components with \emph{invariants}
that describe each component's routing behavior.
Given the interfaces of a component's neighbors, we can
verify that the component respects its own interface.
When the interfaces imply a useful property, \EG reachability or access control,
we can conclude that the monolithic network satisfies that property.

We propose \sysname{}, the first modular technique with
\emph{abstract network interfaces to verify a wide range of properties} (including route reachability).
Kirigami~\cite{kirigami} proposed an architecture for modular control plane verification,
but restricted its interfaces to only \emph{exact} routes.
Lightyear~\cite{lightyear} presented an alternative verification technique with more expressive interfaces,
but can only check that a network never receives a route (\EG for access control properties)
--- it cannot check reachability, a keen property of interest.

\paragraph{A temporal model}
The basis of \sysname{}'s approach is a \emph{temporal model of network execution},
where we reason over the states of nodes \emph{at all times}.
This model came as a surprise to us: one branch of prior work, starting with Minesweeper~\cite{minesweeper},
sought to avoid the burden of reasoning over all transient states of the network
by focusing on the \emph{stable states} of the routing protocol once routing converges.
Unfortunately, a naïve combination
of modular reasoning and Minesweeper-style
analysis of stable states \emph{is unsound}.
We discovered that the best way to recover soundness, while maintaining the system's generality,
is to move to a temporal model.

\begin{wrapfigure}{R}{6cm}
	\centering
	\begin{tikzpicture}
		\begin{axis}[
				thick,
				grid=major,
				height=4cm,
				width=5.5cm,
				xlabel=Topology Size (Nodes),
				ylabel={Verification time},
				scaled ticks=false,
				y unit=\si{\second},
				cycle list name=cute,
				legend style={
						at={(0.98,0.75)},
						anchor=north east,
					}
			]
			\addplot table[x expr=\thisrow{k}^2*1.2, y expr=\thisrow{tk}/1000]{\htbl};
			\addlegendentry{Modular}
			\addplot table[x expr=\thisrow{k}^2*1.2, y expr=\thisrow{ms}/1000]{\htbl};
			\addlegendentry{Monolithic}
			\addplot[timeout line] {7200.0} node [timeout node] {timeout};
		\end{axis}
	\end{tikzpicture}
	\caption{Verification time comparison between \sysname{} and Minesweeper-style verification.}
	\label{fig:hijack-preview}
	\Description{A graph with two trend lines, one for modular verification in blue and another for monolithic verification in orange.
      The x-axis is the size of a network topology in terms of nodes, and ranges from 20 to 2000 nodes.
      The y-axis is the time taken by verification in seconds for the benchmark, and ranges from 0 to 8000 seconds.
      There are data points at gradually increasing intervals, starting at 20 nodes, and then going to 80, 180, 320, 500,
      720, 980, 1280, 1620 and finally 2000 nodes.
      The modular verification line grows very gradually with the size of the network and hugs close to 0 on the y-axis.
	  The monolithic verification line quickly reaches the 2-hour (7200 second) timeout at 80 nodes.}
\end{wrapfigure}
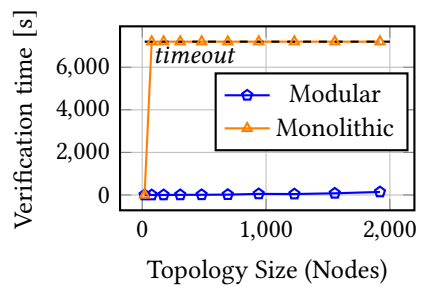

This temporal model appears to ask the verification engine to do a lot more work:
the system must verify that all the messages produced
\emph{at all times} are consistent with a user-supplied
interface for each network component.
Nevertheless, because reasoning is modular, ensuring individual problems are small,
the system scales with the size of the largest \emph{component} rather than the size of the network.
This modular reasoning is general and any symbolic method
(\EG symbolic simulation, model checking) could use it to verify individual components.
We use a Satisfiability Modulo Theories (SMT)-based method in this work~\cite{smt}.
As a preview of modularity's benefits, Figure~\ref{fig:hijack-preview} shows the time
it takes \sysname{} to verify connectivity for
variable-sized fattree topologies~\cite{fattree} with external route announcements
using the eBGP routing protocol, compared with a Minesweeper-style network-wide stable paths encoding.

\sysname{} does require more work of users than monolithic, non-modular systems:
users must supply interfaces that characterize the routes
each network component may generate at each time.
Still, these interfaces, once constructed, provide the
typical benefits of interfaces in any software engineering context.
First, they localize exactly where an error occurs:
if a component is not consistent with its interface,
then one must search \emph{only} that component for the mistake,
and a counterexample from the SMT solver can help pinpoint it.
Second, router configurations change rapidly, and these changes are often
the source of network-wide problems~\cite{dna}.
Well-defined interfaces will be stable over time.
As users update their configurations, they may easily recheck them against
the stable, local interface for problems.

Inspired by temporal logic~\cite{pnueli-modular}, we
developed a simple language to help users specify their interfaces.
Through this language, users may state that they expect to see certain sets of routes
\emph{always}, \emph{eventually} (by some specified time $t$, to be more precise),
or \emph{until} (some approximate specified time).
Moreover, the interface language allows users to write abstract specifications that
need not characterize irrelevant features of routes, and instead
only what is necessary to prove a desired property.
For instance, a user might specify a reachability property simply by
stating a node must ``eventually receive some route \new{at time $t$},'' without saying
which route it must receive.
\new{Our formal model is based on a synchronous semantics of time,
  where nodes receive updates in lockstep.
  As discussed in prior work~\cite{daggitt2018asynchronous},
  this simplifies reasoning over the routing behavior of networks
  which converge to unique solutions.
  One may extend this model to consider a bounded number of steps of \emph{delay}
  at the cost of increasing the complexity of our invariants.}

To summarize, the key contributions of this paper are:
\begin{itemize}
	\item We demonstrate in depth why a natural, but naïve
	      modular control plane analysis based on an analysis of
	      stable states is unsound (\S\ref{sec:overview}).
	\item We develop a new theory for modular control plane analysis
	      based on time (\S\ref{sec:theory}). We prove it sound with respect
        to \new{the semantics of a network simulator},
        \new{and complete with respect to the closed network semantics
        (starting from fixed initial values)}.
        This theory is general, and can verify individual components
        using any verification method.
  \item We define an SMT-based verification procedure to reason
        about all possible routes at all times,
        \new{which can analyze networks with symbolic representations of, \EG
        external announcements or destination routers.} (\S\ref{sec:algorithms})
	\item We design and implement a new, modular control plane
	      verification tool, \sysname{}, based
	      directly on this procedure (\S\ref{sec:implementation}).
        We evaluate \sysname{} and check a variety of policies
	      at individual nodes in hundreds of milliseconds.
	      Thanks to its embarrassingly parallel modular procedure,
	      \sysname{} scales to networks with thousands of nodes (\S\ref{sec:evaluation}).
\end{itemize}


\section{Key Ideas}
\label{sec:overview}

This section introduces the stable routing model of network control
planes, which serves as a foundation for many past network verification tools~\cite{minesweeper,bonsai,shapeshifter,plankton}.
It illustrates in depth why naïvely adopting this model for modular verification
is unsound.
It then introduces a new temporal model for control
plane verification and provides the intuition for why the revised
model is superior.
This section is long but contains a substantial payoff:
the essence of why a sound and general modular control plane analysis
should be based off a temporal model of control plane behavior.

\subsection{Background}
\label{sec:background}
To determine how to deliver traffic between two endpoints, routers
(also called nodes) run distributed
routing protocols such as BGP~\cite{bgp}, OSPF~\cite{ospf}, RIP~\cite{rip}, or ISIS~\cite{isis}.
Each node participating in a
protocol receives \emph{messages} (also called \emph{routes}) from its neighboring nodes.
After receiving routes from its neighbors, a node will select its
``best'' route---the route it will use to forward traffic.
Different protocols use different metrics to compare routes and select
the best among those received.
For instance, RIP compares hop count; OSPF uses the shortest weighted-length path;
and BGP uses a complex, user-configurable combination of metrics.
Finally, each router sends its chosen route to its
neighbors, possibly modifying the route along the way
(for instance, by prepending its identifier to the path represented by the route).

\paragraph{Routing algebras}
Routing algebras~\cite{metarouting,routingalgebra} are abstract models that capture the similarities between
different distributed routing protocols.
Prior work on control plane verification~\cite{stable-paths,minesweeper,nv}
uses similar abstract models to formalize route computation.
We adopt this standard abstract model of routing protocols, which specifies the following components.
\begin{itemize}
	\item A directed graph $G$ that defines the network topology's nodes ($V$) and edges ($E$).
	      We use lowercase
	      letters ($u$, $v$, $w$, \ETC) for nodes and pairs ($uv$) to indicate directed edges.
	\item A set $S$ of \emph{routes} that communicate routing information between nodes.
        \new{Routes abstract the routing announcements and metarouting information used in different routing
        protocols.
        Depending on the problem under consideration, $S$ may be the set of Booleans $\mathbb{B}$
        or natural numbers $\mathbb{N}$ (\EG checking reachability or path length properties),
        a set of flags (\EG checking access control), or a record with multiple fields (more complex policies and/or properties).}
	\item An initialization function $\init{}$ that
	      provides an initial route $\init{v} \in S$ for each node $v$.
	\item \new{A function $\TT$ that maps edges to transfer functions.}
	      Each transfer function $\TT(e) = \T_{e}$ transforms routes as they traverse the edge $e$.
	\item A binary \new{associative and commutative} function $\M$ (\AKA \emph{merge} or
	      the \emph{selection function})
	      selects the best route between two options.
\end{itemize}

\paragraph{An idealized example}
Many large cloud providers deploy data center networks to scale
up their compute capacity.
They connect those data centers to each other and the rest of the Internet
via a wide-area network (WAN).
To illustrate the challenges of modular network verification, we
will explore verification of an idealized cloud provider network
with WAN and data center components.
Figure~\ref{fig:cloud-example} presents a highly abstracted view of our network's topology.
The data center network contains
routers \fn{d} and \fn{e} where \fn{d} connects to the corporate
WAN and \fn{e} connects to data center servers.
The WAN consists of routers \fn{w} and \fn{v}.
Router \fn{v} connects
to the data center as well as to a neighboring network \fn{n},
which is not controlled by our cloud provider.%
\footnote{Any of the edges could be
	bi-directional, allowing routes to pass in both directions,
	but for pedagogic reasons we strip down the example to the barest
	minimum, retaining edges that flow from left-to-right except
	for at \fn{v} and \fn{d} where routes may flow back and forth.}

\begin{figure}
  \centering
  \begin{minipage}{0.5\linewidth}
	\begin{tikzpicture}
		\Vertex[x=1.6,y=.3,fontsize=\normalsize,size=2.5,RGB,color={250,250,250}]{WAN}
		\Vertex[x=4.5,y=.3,fontsize=\normalsize,size=2.5,RGB,color={250,250,250}]{DC}
		\Vertex[x=0,y=-.5,label=$n$,fontsize=\normalsize,size=.6,color=nodeorange]{N}
		\Vertex[x=1.1,y=.6,label=$w$,fontsize=\normalsize,size=.6]{W}
		\Vertex[x=2.3,label=$v$,fontsize=\normalsize,size=.6]{V}
		\Vertex[x=3.7,label=$d$,fontsize=\normalsize,size=.6]{D}
		\Vertex[x=5.3,label=$e$,fontsize=\normalsize,size=.6]{E}
		\Text[x=0,y=-1.2,fontsize=\small]{\textbf{Neighbor}}
		\Text[x=1.7,y=-1.2,fontsize=\small]{\textbf{WAN}}
		\Text[x=4.6,y=-1.2,fontsize=\small]{\textbf{Data Center}}
		\Edge[Direct,NotInBG,label=filter,fontsize=\footnotesize](N)(V)
		\Edge[Direct,NotInBG,label=tag,fontsize=\footnotesize](W)(V)
		\Edge[Direct,NotInBG,bend=20](V)(D)
		\Edge[Direct,NotInBG,bend=20](D)(V)
		\Edge[Direct,NotInBG,label=allow,fontsize=\footnotesize](D)(E)
	\end{tikzpicture}
  \end{minipage}
  \begin{minipage}{0.4\linewidth}
    \centering
    {\textbf{Routing policies:}}
    \begin{align*}
      \mathsf{filter}&: \text{drop all routes} \\
      \mathsf{tag}&: \text{tag routes internal} \\
      \mathsf{allow}&: \text{drop external routes}
    \end{align*}
  \end{minipage}
  \caption{Our idealized example cloud provider network.}
  \label{fig:cloud-example}
  \Description{An image representing a cloud provider network as a labelled graph.
    Five nodes are visible:
    an external neighbor \fn{n}, two nodes \fn{w} and \fn{v} in a wide-area network,
    and two nodes \fn{d} and \fn{e} in a data center.
    Links are shown with directed edges.
    External neighbor node \fn{n} has an edge to WAN node \fn{v}, which filters all routes;
    WAN node \fn{w} also has an edge to node \fn{v}, which tags routes from \fn{w} as internal;
    node \fn{v} has an edge to data center node \fn{d};
    and finally node \fn{d} has two edges, one back to \node v and one to data center \node e.
    The \fn{de} edge drops routes that are not tagged as internal.}
\end{figure}

The default routing policy uses shortest-paths.
However, in addition,
the network administrators want \fn{e}
to be reachable from all cloud-provider-owned
devices (\IE \fn{w}, \fn{v}, \fn{d}), but not to be reachable from
outsiders (\IE \fn{n}).
They intend to enforce this property by tagging all routes originating
from their network ($w$) as ``internal'' (\EG using BGP community tags~\cite{bgp-communities}) and allowing
those routes to traverse the \fn{de} edge.
Doing so should allow \fn{e} to communicate with internal machines
but not external machines.
Furthermore, to protect nodes from outside interference,
the cloud provider applies route filters to external peers to drop
erroneous advertised routes that may ``hijack''~\cite{rcc} internal routing.

\paragraph{Modelling the example}
To model our example network,
we define the network topology as the graph $G$ pictured earlier.  We assume all
routers participate in an idealized variant of eBGP~\cite{bgp}, which is commonly used in both wide-area networks and data centers~\cite{bgp-in-dc}.
\new{We abstract away some of the fields of eBGP routing announcements to define the set of routes $S$ as
records with 3 fields:} 
\begin{enumerate*}[label=\emph{(\roman*)}]
	\item an integer ``local preference'' that lets users
	overwrite default preferences,
	\item an integer path length, and
	\item a boolean tag field that is set to \TR{} if a route comes
	from an internal source and false otherwise.
\end{enumerate*}
$S$ also includes \nullm{}, a message that indicates \emph{absence} of a route.

Let's consider what happens when starting with a specific route at WAN node \fn{w},
$\Rec{100}{0}{\FL}$
(local preference 100, path length of 0, not tagged internal).
The $\init{}$ function assigns \fn{w} that route,
and assigns the \nullm{} route to all other nodes.

The transfer function $\T_{e}$ increments the length field of every route by one across every edge $e$.
In addition, edge \fn{wv} sets the internal tag field to \TR{} and edge \fn{nv} drops all routes
(transforms them into $\nullm{}$).
Finally, edge \fn{de} drops all routes not tagged internal/true.

The merge function $\M$ always prefers some route
over the \nullm{} route, and prefers routes with
higher local preference over lower local preference.  If the local
preference is the same, it chooses a route with a shorter path length.
\M{} ignores the tag field.
For example, \M{} operates as follows:
\[
	\begin{array}{ccccc}
		\Rec{100}{2}{\FL} & \M & \nullm{}          & = & \Rec{100}{2}{\FL} \\
		\Rec{100}{2}{\FL} & \M & \Rec{200}{5}{\TR} & = & \Rec{200}{5}{\TR} \\
		\Rec{200}{2}{\FL} & \M & \Rec{200}{5}{\TR} & = & \Rec{200}{2}{\FL} \\
	\end{array}
\]

\paragraph{Network simulation}
A \emph{state} of a network is a mapping from nodes to the ``best routes'' they have computed so far.
One may simulate a network by starting in its initial state and
repeatedly computing new states (\IE new ``best routes'' for particular nodes).
Well-behaved networks eventually converge to \emph{stable states}
where no node can compute a better route, given the routes
provided by its neighbors.

To compute a new best route at a particular node, say \fn{v},
we apply the \T{} function to each best route
computed so far at its neighbors \fn{w}, \fn{n}, and \fn{d}, and then select
the best route among the results and the initial value at \fn{v},
using the merge ($\M$) function.  More precisely:
\begin{equation*}
	\mathit{v_{new}} =
	\T_{wv} (\mathit{w_{old}}) \M
	\T_{nv} (\mathit{n_{old}}) \M
	\T_{dv} (\mathit{d_{old}}) \M
	\init{v}
\end{equation*}

The table in Figure~\ref{fig:simulation} presents an example
simulation.  At each time step, all nodes compute their best route
given the routes sent by their neighbors at the previous time step.
\new{Our model assumes a synchronous time semantics for simplicity:
  this simulation is hence one possible asynchronous execution.}%
\footnote{
  \new{See \S\ref{sec:algorithms} for a discussion of how we can extend our model
  to consider networks with delay.}}
After time step 3, no node computes a better route---the system
has reached a \emph{stable state}. The picture in Figure~\ref{fig:simulation} annotates each node in the diagram
with the stable route it computes.

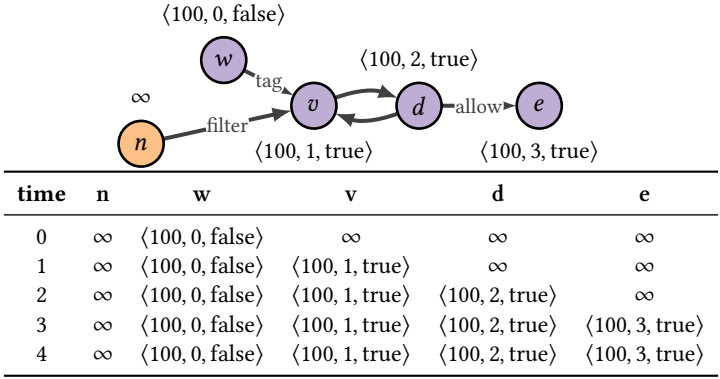
\begin{figure}
  \centering
    \begin{tikzpicture}
		\Vertex[x=0,y=-.5,label=$n$,fontsize=\normalsize,size=.6,color=nodeorange]{N}
		\Vertex[x=1.1,y=.6,label=$w$,fontsize=\normalsize,size=.6]{W}
		\Vertex[x=2.3,label=$v$,fontsize=\normalsize,size=.6]{V}
		\Vertex[x=3.7,label=$d$,fontsize=\normalsize,size=.6]{D}
		\Vertex[x=5.3,label=$e$,fontsize=\normalsize,size=.6]{E}
        \Text[x=0,y=.1,fontsize=\small]{$\nullm{}$}
        \Text[x=1.1,y=1.2,fontsize=\small]{$\Rec{100}{0}{\FL{}}$}
        \Text[x=2.3,y=-.6,fontsize=\small]{$\Rec{100}{1}{\TR{}}$}
        \Text[x=3.7,y=.6,fontsize=\small]{$\Rec{100}{2}{\TR{}}$}
        \Text[x=5.3,y=-.6,fontsize=\small]{$\Rec{100}{3}{\TR{}}$}
        \Edge[Direct,label=filter,fontsize=\footnotesize](N)(V)
        \Edge[Direct,label=tag,fontsize=\footnotesize](W)(V)
        \Edge[Direct,bend=20](V)(D)
        \Edge[Direct,bend=20](D)(V)
		\Edge[Direct,label=allow,fontsize=\footnotesize](D)(E)
    \end{tikzpicture}
    \small
    \begin{tabular}{ cccccc }
        \toprule
        \textbf{time} & $\mathbf{n}$ & $\mathbf{w}$                   & $\mathbf{v}$                  & $\mathbf{d}$                  & $\mathbf{e}$                  \\ \midrule
        $0$           & $\nullm{}$   & $\Rec{100}{0}{\FL{}}$ & $\nullm{}$                    & $\nullm{}$                    & $\nullm{}$                    \\
        $1$           & $\nullm{}$   & $\Rec{100}{0}{\FL{}}$ & $\Rec{100}{1}{\TR{}}$ & $\nullm{}$                    & $\nullm{}$                    \\
        $2$           & $\nullm{}$   & $\Rec{100}{0}{\FL{}}$ & $\Rec{100}{1}{\TR{}}$ & $\Rec{100}{2}{\TR{}}$ & $\nullm{}$                    \\
        $3$           & $\nullm{}$   & $\Rec{100}{0}{\FL{}}$ & $\Rec{100}{1}{\TR{}}$ & $\Rec{100}{2}{\TR{}}$ & $\Rec{100}{3}{\TR{}}$ \\
        $4$           & $\nullm{}$   & $\Rec{100}{0}{\FL{}}$ & $\Rec{100}{1}{\TR{}}$ & $\Rec{100}{2}{\TR{}}$ & $\Rec{100}{3}{\TR{}}$ \\
        \bottomrule
    \end{tabular}
    \normalsize
    \Description{A figure showing the running example network and its converged routes following simulation.
        Node \fn{n} starts with and retains the null route $\nullm{}$.
        Node \fn{w} starts with and retains the route $\Rec{100}{0}{\FL{}}$.
        Node \fn{v} starts with the null route $\nullm{}$ but then acquires a route $\Rec{100}{1}{\FL{}}$ from \fn{w} at time 1.
        Node \fn{d} acquires this route from \fn{v} at time 2, after starting with the null route $\nullm{}$.
        Finally, node \fn{e} acquires \fn{d}'s route at time 3, also after starting with the null route $\nullm{}$.
        At time 4, no node updates their route, so the network converged at the previous time step.}
	\caption{Simulation of the example network for a fixed set of initial routes.
		Node \fn{e} receives a route from \fn{d} since its route is tagged as internal, and the network stabilizes at time 3.}
	\label{fig:simulation}
\end{figure}

\paragraph{Network verification}
Since the edge from \fn{d} to \fn{e} only allows routes tagged internal,
$w$'s route would not reach \fn{e} if \fn{v} were to receive a better route from \fn{n}
(\EG if the route filter from $n$ was implemented incorrectly).
In other words, the simulation demonstrates
that the network correctly operates when \fn{n} sends no route ($\nullm{}$).
But what about other routes? Will $f_{nv}$ filter all routes from \fn{n} correctly?
SMT-based tools like Minesweeper~\cite{minesweeper} and Bagpipe~\cite{bagpipe} can answer such questions
by translating the routing problem into 
constraints for a Satisfiability Modulo Theory (SMT) solver to solve.
\new{An SMT-based encoding of our network could represent any possible external route announcement
  from \fn{n} by representing its initial value with a symbolic variable: the solver can then
  search for a concrete route captured by this variable
  that \emph{violates} our desired property,
  \IE a stable state where \fn{e} never receives a route from \fn{w}.}

\subsection{The Challenge of Modular Verification}
\label{sec:the-challenge}
A system for modular verification will partition a network
into components and verify each component separately, possibly in parallel.
However, since routes computed at a node in one component depend on the routes sent by nodes in neighboring components, each
component must make some assumptions about the routes produced by its neighbors.

\paragraph{Interfaces}
In our case, for simplicity (though this is not necessary),
we place every node in its own component and define
for it an \emph{interface} that attempts to \emph{overapproximate} \new{(or equal)}
the set of routes that the node might produce in a stable state.
The interface for the network as a whole is a function $\An$ from
nodes to sets of routes where $\An(x)$ is the interface for node $x$.

The person attempting to verify the network will supply these interfaces.
Of course, interfaces may be \emph{wrong}---that is, they might not include
some route computed by a simulation (and hence might not be a proper
overapproximation).
Indeed, when there are bugs in the network, the interfaces a user supplies
are likely to be wrong!
The user \emph{expects} the network to behave one way, producing a certain
set of routes, but the network behaves differently due to
an error in its configuration.
A sound modular verification procedure must detect such errors
\new{and indicate if we must strengthen the interface to prove the property}.
On the other hand, a useful modular verification
procedure should allow interfaces to overapproximate the routes
produced, when users find it convenient.
Overapproximations are sound for verifying properties over all routing behaviors
of a network,
and they often simplify reasoning,
allowing users to think more abstractly.

Throughout the paper, we use predicates $\predicate$ to define interfaces,
where $\predicate$ stands in for the set of routes $\{ s ~\vert~ s \in S, \predicate(s) \}$.
Returning to our running example,
one might define the interface for \fn{w}
using the predicate $s.\lp = 100 \wedge s.\len = 0 \wedge \neg s.\bgptag$.
Such an interface would include exactly the one route
generated by \fn{w} in our example: $\Rec{100}{0}{\FL}$.
However, path length is unimportant in the current context;
to avoid thinking about it, a user could instead provide a weaker interface
representing infinitely many possible routes,
such as $s.\lp = 100 \wedge \neg s.\bgptag$.
This interface relieves
the user of having to figure out the exact path length (not so hard
in this simple example, but potentially challenging in an
arbitrary wide-area network), and instead specifies only the
local preference and the tag.
In general, admitting overapproximations make
it possible for users to ignore any features of routing
that are not actually relevant for analyzing the properties of interest.

\paragraph{The strawperson verification procedure}
For a given node \fn{x}, the \emph{component centered at \fn{x}} is the subgraph
of the network that includes node \fn{x} and all edges that end at \fn{x}.
Given a network interface $\An$, our strawperson verification procedure (\SV) will
consider the component centered at each node \fn{x} independently.
Suppose a node \fn{x} has neighbors $n_1, \ldots, n_k$.
For that node \fn{x},
\SV checks that
\begin{equation}\label{eq:sv-check}
  \forall s_1 \in A(n_1), \ldots, \forall s_k \in A(n_k),\
	\T_{n_1x} (s_1) \M
	\cdots \M
	\T_{n_kx} (s_k) \M
	\init{x} \in A(x)
\end{equation}
This check is akin to performing one local step of simulation, checking that
all possible inputs from neighbors produce an output route
satisfying the interface.
We might \emph{hope} that by performing such a check on \emph{all} components independently,
we could guarantee that all nodes converge to stable states satisfying their interfaces.
If that were the case, then we could verify properties by:

\begin{enumerate}
	\item Checking that all components guarantee their interfaces, under the assumption their neighbors do as well; and
	\item Checking that the interfaces imply the network property of interest (\EG reachability, access control, no transit).
\end{enumerate}

\paragraph{The problem: execution interference}
It turns out this simple and natural verification procedure
is unsound: users can supply interfaces that, when analyzed in isolation, satisfy
equation~\eqref{eq:sv-check} above, but \emph{exclude} stable states computed by simulation.
Hence, the second verification step is pointless:
a destination that appears reachable according to an interface may not be;
conversely, a route that appears blocked may not be.

Let us reconsider the running example, where we
assign \fn{w} an initial route with local preference 100, and assume
the external neighbor $n$ can send us any route ($\TR{}$).
A user could provide the interfaces shown in Figure~\ref{fig:bad-interfaces-full}
to falsely conclude that $e$ will not receive a route from $w$.
\begin{figure}[H]
  \centering
	\begin{tikzpicture}
		\Vertex[x=0,y=-.5,label=$n$,fontsize=\normalsize,size=.6,color=nodeorange]{N}
		\Vertex[x=1.1,y=.6,label=$w$,fontsize=\normalsize,size=.6]{W}
		\Vertex[x=2.3,label=$v$,fontsize=\normalsize,size=.6]{V}
		\Vertex[x=3.7,label=$d$,fontsize=\normalsize,size=.6]{D}
		\Vertex[x=5.3,label=$e$,fontsize=\normalsize,size=.6]{E}
		\Text[x=0,y=.1,fontsize=\small]{$\TR{}$}
		\Text[x=1.1,y=1.2,fontsize=\small]{$s.\lp=100$}
		\Text[x=2.3,y=-.6,fontsize=\small,color=darkred]{$s.\lp = 200 \wedge \neg s.\bgptag$}
		\Text[x=3.7,y=.6,fontsize=\small,color=darkred]{$s.\lp = 200 \wedge \neg s.\bgptag$}
		\Text[x=5.3,y=-.6,fontsize=\small]{$s = \nullm{}$}
		\Edge[Direct,label=filter,fontsize=\footnotesize](N)(V)
		\Edge[Direct,label=tag,fontsize=\footnotesize](W)(V)
		\Edge[Direct,bend=20](V)(D)
		\Edge[Direct,bend=20](D)(V)
		\Edge[Direct,label=allow,fontsize=\footnotesize](D)(E)
	\end{tikzpicture}
  \caption{Running example with bad interfaces.}
  \label{fig:bad-interfaces-full}
  \Description{
    The running example annotated with bad interfaces that pass the strawperson verification procedure,
    but allow us to conclude that node \fn{e} receives no route from \fn{w}.
    The interfaces are shown above each node.
    Node \fn{n} has the interface $\TR{}$, meaning any route;
    node \fn{w} has the interface $s.\lp = 100$, meaning a route with local preference of 100.
    Nodes \fn{v} and \fn{d} have the interfaces $s.\lp = 200 \wedge \neg s.\bgptag$:
    a route with local preference of 200, and which is not tagged internal.
    Finally, node \fn{e} has the interface $s = \nullm{}$, meaning only the null route.
  }
\end{figure}
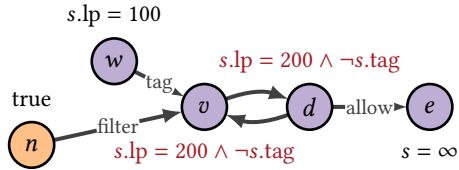
Here, it is easy to check that nodes \fn{n} and \fn{w} satisfy equation~\eqref{eq:sv-check}.
Node \fn{n}'s interface is simply any route.
Node \fn{w}'s route can be any route with a local preference of
100.\footnote{It could be any route ($\TR{}$), as the edge
	$wv$ applies the default preference of 100, but
	for clarity we label the routes at \fn{w} with preference 100.}

The surprise comes at node \fn{v} where its interface \emph{only includes} routes
that satisfy $\neg s.\bgptag$, \IE routes not tagged as internal.
Those routes have $s.\lp = 200$ and may have any path length.
But the route from \fn{w} is tagged \emph{\TR{}} along the
edge $wv$ --- why is such a route erroneously excluded from
\fn{v}'s interface?
We show the component centered at \fn{v} in Figure~\ref{fig:bad-interfaces-v}.
\begin{figure}[H]
	\centering
	\begin{tikzpicture}
		\Vertex[x=0,y=-.5,label=$n$,fontsize=\normalsize,size=.6,color=nodeorange]{N}
		\Vertex[x=1.1,y=.6,label=$w$,fontsize=\normalsize,size=.6]{W}
		\Vertex[x=2.3,label=$v$,fontsize=\normalsize,size=.6]{V}
		\Vertex[x=3.7,label=$d$,fontsize=\normalsize,size=.6]{D}
		\Vertex[x=5.3,label=$e$,fontsize=\normalsize,size=.6,opacity=.2]{E}
		\Text[x=0,y=.1,fontsize=\small]{$\TR{}$}
		\Text[x=1.1,y=1.2,fontsize=\small]{$s.\lp=100$}
		\Text[x=2.3,y=-.6,fontsize=\small,color=darkred]{$s.\lp = 200 \wedge \neg s.\bgptag$}
		\Text[x=3.7,y=.6,fontsize=\small,color=darkred]{$s.\lp = 200 \wedge \neg s.\bgptag$}
		\Text[x=5.3,y=-.6,fontsize=\small,opacity=.2]{$s = \nullm{}$}
		\Edge[Direct,label=filter,fontsize=\footnotesize](N)(V)
		\Edge[Direct,label=tag,fontsize=\footnotesize](W)(V)
		\Edge[Direct,bend=20,opacity=.2](V)(D)
		\Edge[Direct,bend=20](D)(V)
		\Edge[Direct,label=allow,fontsize=\footnotesize,opacity=.2](D)(E)
	\end{tikzpicture}
  \caption{Running example centered on \fn{v}'s component.}
  \label{fig:bad-interfaces-v}
  \Description{
    A close-up on node \fn{v} in the running example annotated with bad interfaces.
    The interfaces are those shown in Figure~\ref{fig:bad-interfaces-full}.
    Node \fn{e} and the edge \fn{vd} are greyed out, focusing attention on \fn{v}
    and the interfaces of its in-neighbors.
  }
\end{figure}
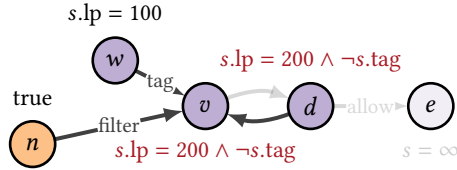
When computing its stable state,
\fn{v} will compare the routes it receives from \fn{w} and \fn{d}:
because all routes from \fn{w} have a local preference set to 100 by $\T_{wv}$,
whereas all routes from \fn{d} have a better local preference of 200,
\fn{v} will always wind up selecting the route from \fn{d} over the route from \fn{w}.

But how then did \fn{d} acquire these preferential routes tagged \FL{}?
Such routes came in turn from \fn{v}'s interface.
Figure~\ref{fig:bad-interfaces-d} shows the component centered at \fn{d}.
\begin{figure}[H]
	\centering
	\begin{tikzpicture}
		\Vertex[x=0,y=-.5,label=$n$,size=.6,color=nodeorange,opacity=.2]{N}
		\Vertex[x=1.1,y=.6,label=$w$,size=.6,opacity=.2]{W}
		\Vertex[x=2.3,label=$v$,size=.6]{V}
		\Vertex[x=3.7,label=$d$,size=.6]{D}
		\Vertex[x=5.3,label=$e$,size=.6,opacity=.2]{E}
		\Text[x=0,y=.1,fontsize=\small,opacity=.2]{$\TR{}$}
		\Text[x=1.1,y=1.2,fontsize=\small,opacity=.2]{$s.\lp=100$}
		\Text[x=2.3,y=-.6,fontsize=\small,color=darkred]{$s.\lp = 200 \wedge \neg s.\bgptag$}
		\Text[x=3.7,y=.6,fontsize=\small,color=darkred]{$s.\lp = 200 \wedge \neg s.\bgptag$}
		\Text[x=5.3,y=-.6,fontsize=\small,opacity=.2]{$s = \nullm{}$}
		\Edge[Direct,label=filter,fontsize=\footnotesize,opacity=.2](N)(V)
		\Edge[Direct,label=tag,fontsize=\footnotesize,opacity=.2](W)(V)
		\Edge[Direct,bend=20](V)(D)
		\Edge[Direct,bend=20,opacity=.2](D)(V)
		\Edge[Direct,label=allow,fontsize=\footnotesize,opacity=.2](D)(E)
	\end{tikzpicture}
  \caption{Running example centered on \fn{d}'s component.}
  \label{fig:bad-interfaces-d}
  \Description{
    A close-up on node \fn{d} in the running example annotated with bad interfaces.
    The interfaces are those shown in Figure~\ref{fig:bad-interfaces-full}.
    Nodes \fn{n}, \fn{w} and \fn{e}, and the edge \fn{dv} are greyed out, focusing attention on \fn{d}
    and the interfaces of its in-neighbor \fn{v}.
  }
\end{figure}
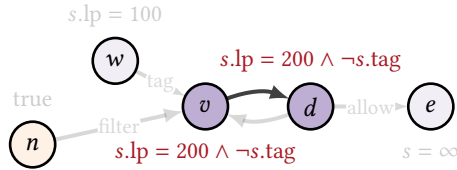
What has happened is that \fn{v} transmits its spurious routes
to \fn{d}, enabling \fn{d} to justify its own spurious routes.
\fn{d} transmits these back again to \fn{v},
where \fn{d}'s routes interact with the legitimate routes from \fn{w}.
Since \fn{w}'s routes have lower local preference, \fn{v} discards them during
computation of stable states.
\new{In a nutshell, our interface proposed routes that do \emph{not} soundly overapproximate the legitimate routes
  from the true simulation, but our verification procedure \emph{accepted} this bad interface
  as it circularly justified itself at \fn{v} and \fn{d}.}
How might we prevent this \emph{execution interference}?

\paragraph{Other approaches}
We can modify this verification procedure to make it sound,
but these solutions will limit the verification procedure's power
or the expressiveness of the properties it can prove.

One approach is to limit every interface to \emph{exactly one route}.
Doing so avoids introducing any imaginary executions in the first place.
Kirigami~\cite{kirigami} takes this approach, but the cost is that a user
analyzing their network must know \emph{exactly} which routes appear at which locations.
Computing routes exactly can be difficult in practice,
and would seem unnecessary if all one cares about is a high-level property such as reachability.
Moreover, it makes the interfaces brittle in the face of change---any change in network configuration likely necessitates a change in interface.
A superior system would allow operators to define \emph{durable} and \emph{abstract} interfaces that imply key properties,
and to check configuration updates against those interfaces.

Another approach is to limit the set of properties that the system can check to only those that say what \emph{does not happen} in the network rather than what does happen.
This is the approach Lightyear~\cite{lightyear} takes.
For instance, Lightyear can check that node \fn{a} will \emph{not} be able to reach node \fn{b}, but
not that \fn{a} and \fn{b} will have connectivity --- a common requirement in networks.

A final approach is to \emph{statically order} the components,
and verify each component according to this ordering, using no information from
the not-yet-verified components.
\new{By ordering \fn{v} before \fn{d}, we would need to satisfy \fn{v}'s invariant
  without routes from \fn{d} (treating \fn{d}'s invariant as \FL{}) ---
  this would fail for the bad invariant $s.\lp = 200 \wedge \neg s.\bgptag$
  using only routes from \fn{n} and \fn{w}.}
In practice, this is still unnecessarily conservative.
The running example is overly simple as it shows routes propagated through a network
in a single direction from left to right.
In realistic networks, multiple destinations may broadcast routes in multiple
directions at once.
In such situations, there
may be no way to order the components, and verification may not be possible.

\subsection{The Solution: A Temporal Model}
\label{sec:the-solution}

Our key insight is to change the model: rather than focus exclusively on the final
stable states of a system, as a Minesweeper-style verifier would,
we ensure that the model preserves the \emph{entirety} of every step-by-step execution.
To make this work, we need to add information to the model: a notion of \emph{logical time}.
By associating every route with the time at which a node computes it, we can
\begin{enumerate*}[label=\emph{(\roman*)}]
	\item ensure that \emph{all} routes at a particular time are properly considered, and their executions extended a time step, and
	\item ensure that we avoid collisions between routes computed at different times.
\end{enumerate*}

To verify such routing systems modularly, we once again must specify interfaces,
but this time the interface for each node will specify the set of routes
that may appear \emph{at any time}.
We write this now as an interface $\Anno{x}{t}$ that takes both a node $x$ \emph{and a time $t$}
and returns an overapproximation of the set of routes that may appear at $x$ \emph{at that time $t$}.
To check the interfaces, we use a verification procedure structured inductively
with respect to time, as follows:
\begin{itemize}
	\item At every node \fn{x}, check $\init{x}$ is included in $\Anno{x}{0}$
	\item Consider each node \fn{x} with neighbors $n_1, \ldots, n_k$.
	      At time $t+1$, check that merging any combination of routes
	      $s_1 \in \Anno{n_1}{t}, \ldots, s_k \in \Anno{n_k}{t}$ from neighbors' interfaces
	      at time $t$ produces a route in $\Anno{x}{t+1}$:
	      \begin{equation}
		      \T_{n_1x} (s_1) \M \cdots \M \T_{n_kx} (s_k) \M \init{x} \in \Anno{x}{t+1}
	      \end{equation}
\end{itemize}
Because this procedure is structured inductively, we can prove,
by induction on time, that all states at all times are included
in their respective interfaces---the procedure is \emph{sound}.

\new{For brevity, we specify our interfaces using \emph{temporal operators}.
  These operators are functions that take a time $t$ as an argument, compare it
  to an explicit time variable $\tau$, and return a predicate.}
We write $\globally(\PP)$ (``globally $\PP$'') when a node's interface includes
the routes that satisfy predicate $\PP$ for all times $t$.
\new{We write $\PP_1 \until{\tau} \tpop_2$ (``$\PP_{1}$ until $\tpop_{2}$'')} when a
node may have routes satisfying $\PP_1$ until time $\tau{-}1$ and
\new{operator $\tpop_2(\tau)$} holds afterwards.
Finally, we write \new{$\finally{\tau}(\tpop)$ (``finally $\tpop$'')}
to mean that eventually at time $\tau$ routes start satisfying \new{$\tpop(\tau)$}.

\paragraph{Verifying correct interfaces}
\new{Figure~\ref{fig:example-weak}} below presents an interface we may verify with this model.
\begin{figure}[H]
	\centering
	\begin{tikzpicture}
		\Vertex[x=0,y=-.5,label=$n$,fontsize=\normalsize,size=.6,color=nodeorange]{N}
		\Vertex[x=1.1,y=.6,label=$w$,fontsize=\normalsize,size=.6]{W}
		\Vertex[x=2.3,label=$v$,fontsize=\normalsize,size=.6]{V}
		\Vertex[x=3.7,label=$d$,fontsize=\normalsize,size=.6]{D}
		\Vertex[x=5.3,label=$e$,fontsize=\normalsize,size=.6]{E}
		\Text[x=0,y=.1,fontsize=\small]{$\globally(\TR)$}
		\Text[x=1.1,y=1.2,fontsize=\small]{$\globally(s.\mathrm{lp} = 100)$}
		\Text[x=2.3,y=-.6,fontsize=\small]{$\globally(s = \nullm{} \vee s.\bgptag)$}
		\Text[x=3.7,y=.6,fontsize=\small]{$\globally(s = \nullm{} \vee s.\bgptag)$}
		\Text[x=5.3,y=-.6,fontsize=\small]{$\globally(s = \nullm{} \vee s.\bgptag)$}
		\Edge[Direct,label=filter,fontsize=\footnotesize](N)(V)
		\Edge[Direct,label=tag,fontsize=\footnotesize](W)(V)
		\Edge[Direct,bend=20](V)(D)
		\Edge[Direct,bend=20](D)(V)
		\Edge[Direct,label=allow,fontsize=\footnotesize](D)(E)
	\end{tikzpicture}
  \caption{Running example with interfaces proving that if \fn{e} has a route, it is tagged.}
  \label{fig:example-weak}
  \Description{
    The running example annotated with temporal interfaces that allow us to conclude
    that if node \fn{e} has a route, then that route is tagged.
    The interfaces are shown above each node.
    Node \fn{n} has the interface $\globally(\TR{})$, meaning any route at all times;
    node \fn{w} has the interface $\globally(s.\lp = 100)$, meaning a route with local preference of 100 at all times.
    Nodes \fn{v}, \fn{d} and \fn{e} have the interfaces $\globally(s = \nullm{} \vee s.\bgptag)$:
    at all times, they have a null route, or a route tagged internal.
  }
\end{figure}
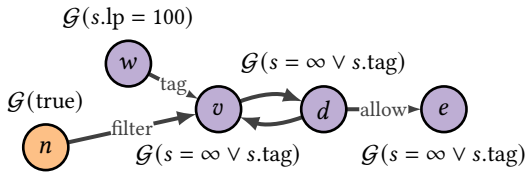
We again assume that \fn{n} sends any route at any time, denoted by the interface: $\globally(\TR)$.
We assume \fn{w} has a route with default local preference: $\globally(s.\lp = 100)$.
The interesting part is at nodes \fn{v} and \fn{d} where the interfaces state
that there is always no route (\EG at time 0), or a tagged route:
$\globally(s = \nullm{} \vee s.\bgptag)$.
We can then prove a weak property about node \fn{e}:
\emph{if it receives a route}, then the route will be tagged internal.
We can prove node \fn{v}'s route satisfies its interface since $\T_{wv}$ tags routes on import from node \fn{w},
routes from \fn{n} are correctly dropped, and routes from \fn{d} must also have a tag per its interface.
In fact, all the nodes satisfy their interface given their neighbors' interfaces.

\paragraph{Proving reachability}
Figure~\ref{fig:example-weak}'s interfaces were too weak to prove that \fn{w} can reach \fn{e}.
The problem is that they reason about \emph{all} times (\IE from time 0 onward),
yet \fn{e} will only \emph{eventually} have a route from \fn{w} at some time in the future.
Consider now the stronger interfaces shown \new{in Figure~\ref{fig:example-reach}}:
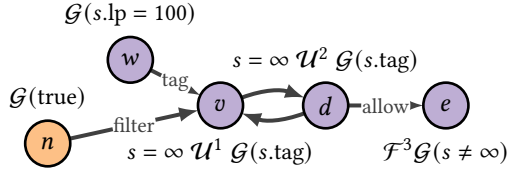
\begin{figure}[H]
	\centering
	\begin{tikzpicture}
		\Vertex[x=0,y=-.5,label=$n$,size=.6,RGB,color={253,192,134}]{N}
		\Vertex[x=1.1,y=.6,label=$w$,size=.6]{W}
		\Vertex[x=2.3,label=$v$,size=.6]{V}
		\Vertex[x=3.7,label=$d$,size=.6]{D}
		\Vertex[x=5.3,label=$e$,size=.6]{E}
		\Text[x=0,y=.1,fontsize=\small]{$ \globally(\TR)$}
		\Text[x=1.1,y=1.2,fontsize=\small]{$ \globally(s.\lp = 100)$}
		\Text[x=2.3,y=-.6,fontsize=\small]{$s = \nullm \until{1}\globally(s.\bgptag)$}
		\Text[x=3.7,y=.6,fontsize=\small]{$s = \nullm \until{2}\globally(s.\bgptag)$}
		\Text[x=5.3,y=-.6,fontsize=\small]{$\finally{3}\globally(s \neq \nullm{})$}
		\Edge[Direct,label=filter,fontsize=\footnotesize](N)(V)
		\Edge[Direct,label=tag,fontsize=\footnotesize](W)(V)
		\Edge[Direct,bend=20](V)(D)
		\Edge[Direct,bend=20](D)(V)
		\Edge[Direct,label=allow,fontsize=\footnotesize](D)(E)
	\end{tikzpicture}
  \caption{Running example with interfaces proving \fn{e} can reach \fn{w}.}
  \label{fig:example-reach}
  \Description{
    The running example annotated with temporal interfaces that allow us to conclude
    that node \fn{e} can reach \fn{w}.
    The interfaces are shown above each node.
    Node \fn{n} has the interface $\globally(\TR{})$, meaning any route at all times;
    node \fn{w} has the interface $\globally(s.\lp = 100)$, meaning a route with local preference of 100 at all times.
    Node \fn{v} has the interface $s = \nullm \until{1} \globally(s.\bgptag)$, meaning it has a null route
    until time step 1, and then from time step 1 onward the route is tagged.
    Node \fn{d} has the interface $s = \nullm \until{2} \globally(s.\bgptag)$: its route is null
    until time step 2, and then the route is tagged thereafter.
    Node \fn{e} has the interface $\finally{3}\globally(s \neq \nullm{})$, meaning from time step 3 onward
    its route is not null.
  }
\end{figure}

As before, we allow \fn{n} and \fn{w} to send any route.
However, now nodes \fn{v} and \fn{d} declare that they will not have a route \emph{until}
a specified (logical) time, at which point they receive a tagged route.
We give precise witness times for \fn{v} and \fn{d}'s interfaces, as otherwise
\fn{v} could give \fn{d} a non-null route (or vice-versa) that would violate the interface
before its witness time.
\fn{e}'s interface simply requires that \fn{e} receives some
route at the witness time (allowing arbitrary routes before the witness time).
These interfaces are sufficient to prove that \fn{e} will eventually
receive a route to \fn{w}, since \fn{d} will eventually have a route tagged as internal,
and hence \fn{e} will allow it.

\paragraph{Debugging erroneous interfaces}
Let us revisit the example where a user gave unsound interfaces
using spurious routes with local preference 200.
Figure~\ref{fig:example-bad-timed} presents the equivalent temporal interfaces.
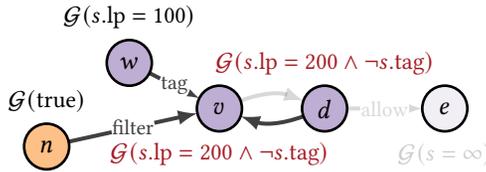
\begin{figure}[H]
	\centering
	\begin{tikzpicture}
		\Vertex[x=0,y=-.5,label=$n$,size=.6,RGB,color={253,192,134}]{N}
		\Vertex[x=1.1,y=.6,label=$w$,size=.6]{W}
		\Vertex[x=2.3,label=$v$,size=.6]{V}
		\Vertex[x=3.7,label=$d$,size=.6]{D}
		\Vertex[x=5.3,label=$e$,size=.6,opacity=.2]{E}
		\Text[x=0,y=.1,fontsize=\small]{$\globally(\TR{})$}
		\Text[x=1.1,y=1.2,fontsize=\small]{$\globally(s.\lp=100)$}
		\Text[x=2.3,y=-.6,fontsize=\small,color=darkred]{$\globally(s.\lp = 200 \wedge \neg s.\bgptag)$}
		\Text[x=3.7,y=.6,fontsize=\small,color=darkred]{$\globally(s.\lp = 200 \wedge \neg s.\bgptag)$}
		\Text[x=5.3,y=-.6,fontsize=\small,opacity=.2]{$\globally(s = \nullm{})$}
		\Edge[Direct,label=filter,fontsize=\footnotesize](N)(V)
		\Edge[Direct,label=tag,fontsize=\footnotesize](W)(V)
		\Edge[Direct,bend=20,opacity=.2](V)(D)
		\Edge[Direct,bend=20](D)(V)
		\Edge[Direct,label=allow,fontsize=\footnotesize,opacity=.2](D)(E)
	\end{tikzpicture}
  \caption{Running example with bad temporal interfaces.}
  \label{fig:example-bad-timed}
  \Description{
    The running example annotated with bad temporal interfaces lifted from our earlier
    execution interference example in Figure~\ref{fig:bad-interfaces-full}.
    Each interface is now wrapped in a $\globally$ temporal operator.
  }
\end{figure}

Unlike before, the verification procedure detects an error: the interfaces at nodes
\fn{v} and \fn{d} do not include the initial route $\nullm{}$ at time 0.
As a result, the user will receive a counterexample for time $t = 0$ when verifying \fn{v} or \fn{d}.
Suppose our imaginative user tries to circumvent this issue by also including the initial route in the interfaces
for \fn{v} and \fn{d} with the interface:
\begin{equation*}
	\globally\big( (s.\lp = 200 \wedge \neg s.\bgptag) \vee (s = \nullm{}) \big)
\end{equation*}
However, doing so merely pushes the problem ``one step forward in time''---there is no way to circumvent our temporal analysis.
If \fn{d}'s route may be $\nullm{}$, \fn{v}'s interface must also consider what routes it selects when that is the case,
including tagged routes such as $\Rec{100}{1}{\TR}$.
The user might receive a counterexample at time $t = 1$ where \fn{v}'s route is the following:
\begin{equation*}
	\T_{wv} (\Rec{100}{0}{\FL}) \M \T_{nv}(\nullm{}) \M \T_{dv}(\nullm{}) = \Rec{100}{1}{\TR}
\end{equation*}
where $\Anno{v}{1}$ does not contain the result $\Rec{100}{1}{\TR}$.
This counterexample reveals the fact that there is an error
in either the specification (as in this case) or the configuration
(\EG if a buggy configuration tagged routes from \fn{w} \FL rather than \TR as expected).

\paragraph{Properties and ghost state}
\new{Although our modular properties reference node-local routes,
  we can verify many end-to-end control plane properties using \emph{ghost state}.}
Users may model routes with additional ``ghost'' fields (\CF ghost fields in Dafny~\cite{dafny})
that play no role in a protocol's routing behavior, yet may capture end-to-end properties.
For instance, suppose we added a boolean ``ghost'' field ``fromw''
to indicate if a route originated from node \fn{w} \new{(see Figure~\ref{fig:example-ghost})}.
We assume this field is initially true at \fn{w}, false at all other nodes,
and that transfer functions preserve its value.
With this addition, we can now check that \fn{e} receives a route from \fn{w} and no other node.
\begin{figure}[H]
  \centering
	\begin{tikzpicture}
		\Vertex[x=0,y=-.5,label=$n$,size=.6,color=nodeorange]{N}
		\Vertex[x=1.1,y=.6,label=$w$,size=.6]{W}
		\Vertex[x=2.3,label=$v$,size=.6]{V}
		\Vertex[x=3.7,label=$d$,size=.6]{D}
		\Vertex[x=5.3,label=$e$,size=.6]{E}
		\Text[x=0,y=.1,fontsize=\small]{$ \globally(\neg s.\mathrm{fromw})$}
		\Text[x=1.1,y=1.2,fontsize=\small]{$ \globally(s.\lp = 100 \wedge s.\mathrm{fromw})$}
		\Text[x=2.3,y=-.6,fontsize=\small]{$s = \nullm \until{1}\globally(s.\bgptag \wedge s.\mathrm{fromw})$}
		\Text[x=3.8,y=.6,fontsize=\small]{$s = \nullm \until{2}\globally(s.\bgptag \wedge s.\mathrm{fromw})$}
		\Text[x=5.45,y=-.6,fontsize=\small]{$\finally{3}\globally(s.\mathrm{fromw})$}
		\Edge[Direct,label=filter,fontsize=\footnotesize](N)(V)
		\Edge[Direct,label=tag,fontsize=\footnotesize](W)(V)
		\Edge[Direct,bend=20](V)(D)
		\Edge[Direct,bend=20](D)(V)
		\Edge[Direct,label=allow,fontsize=\footnotesize](D)(E)
	\end{tikzpicture}
  \caption{Running example augmented with a ``fromw'' ghost variable.}
  \label{fig:example-ghost}
  \Description{
    The running example with temporal interfaces allowing us to prove that
    \fn{e} has a route to \fn{w} (as in Figure~\ref{fig:example-reach}),
    but augmented with a ``fromw'' ghost variable.
    The interfaces are shown above each node.
    Node \fn{n} now has the interface $\globally(\neg s.\mathrm{fromw})$, meaning any route \emph{not from \fn{w}} at all times;
    node \fn{w} has the interface $\globally(s.\lp = 100 \wedge s.\mathrm{fromw})$,
    meaning a route with local preference of 100 and from \fn{w} at all times.
    Node \fn{v} has the interface $s = \nullm \until{1} \globally(s.\bgptag \wedge s.\mathrm{fromw})$,
    meaning it has a null route until time step 1, and then from time step 1 onward the route is tagged
    and from \fn{w}.
    Node \fn{d} has the interface $s = \nullm \until{2} \globally(s.\bgptag \wedge s.\mathrm{fromw})$:
    its route is null until time step 2, and then the route is tagged and from \fn{w} thereafter.
    Node \fn{e} has the interface $\finally{3}\globally(s.\mathrm{fromw})$, meaning from time step 3 onward
    its route is from \fn{w}.
  }
\end{figure}
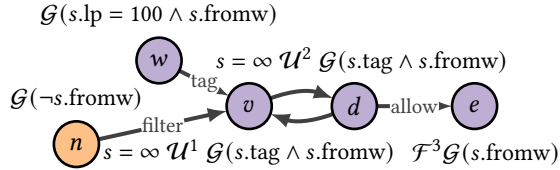

Ghost state allows us to specify and check many network properties;
Table~\ref{fig:example-properties} presents a variety of other possibilities.
\new{For example, to check for routing loops,
  we can use a multi-set of visited nodes (up to $\lvert V \rvert$) and
  mark if a node is ever visited more than once.}
That said, while ghost state is general and flexible, it can only capture information
about the history of a route at a \emph{single} node and is thus not a panacea.
For instance, properties involving the routes at more than one node, such as a formulation of local equivalence~\cite{minesweeper},
where $\st(u)(t) = \st(v)(t)$ for some arbitrary $u,v$ and $t$, is inexpressible using our verifier.
\new{We focus on properties of the network's \emph{control plane}, without considering the data plane forwarding behavior.
  This rules out data plane properties such as load balancing and multipath consistency.}

\begin{table}
	\centering
	\caption{Ghost state for selected example properties.}\label{fig:example-properties}
	\begin{tabular}{ ll }
		\toprule
		\textbf{Property}                       & $\textbf{Added ghost state}$  \\
    \midrule
		reachability to $d$~\cite{batfish}      & \textit{1 bit to mark routes from }$d$                               \\
		isolation~\cite{minesweeper}            & \textit{1 bit per isolation domain}                                  \\
		ordered waypoint~\cite{netplumber}.     & $log_{2}(k)$ \textit{bits for }$k$ \textit{waypoints} \\
		unordered waypoint~\cite{minesweeper}   & $k$ \textit{bits for }$k$ \textit{waypoints}                         \\
    \new{routing loops~\cite{minesweeper}} & \new{\textit{up to }$\lvert V \rvert$ \textit{bits to track visited nodes}} \\
		no-transit~\cite{propane}               & \textit{mark with }$\{\mathrm{peer}, \mathrm{prov}, \mathrm{cust}\}$ \\
		fault tolerance~\cite{minesweeper}      & \textit{up to }$\lvert E \rvert$ \textit{bits to track failed edges} \\
		bounded path length~\cite{nod}          & \textit{integer length field}                                        \\
		\bottomrule
	\end{tabular}
\end{table}

\section{Formal Model with Temporal Invariants}
\label{sec:theory}

Figure~\ref{fig:notation-1} presents the key
definitions and notation needed to formalize our verification procedure.
The notation follows from the previous section,
\EG a network instance $N=(G,S,\init{},\TT,\M)$
contains the key components introduced earlier.
To refer to the route computed by a network simulator at node $v$ at time $t$, we use the
notation $\st(v)(t)$ (defined as before---see Figure~\ref{fig:notation-1}).

\begin{figure}[t]
	\centering
	\begin{minipage}[t]{\linewidth}
		{\textbf{Network instances}\quad\fbox{$N = (G,S,\init{},\TT,\M)$}}
		\[
			\begin{array}{lr}
				G = (V,E)                                          & \text{network topology}              \\
				V                                                  & \text{topology nodes}                \\
				E \subseteq V \times V                             & \text{topology edges}                \\
				S                                                  & \text{set of network routes}         \\
				s \mspace{6mu} \in S                               & \text{a route}                       \\
				\init{} \mspace{6mu} : V \rightarrow S             & \text{node initialization function} \\
				\init{v} \in S                                     & \text{initial route at node }v       \\
				\TT \mspace{.5mu}: E \rightarrow (S \rightarrow S) & \text{edge transfer functions}       \\
				\T_{e} : S \rightarrow S                           & \text{transfer function for edge }e  \\
				\M \mspace{1.9mu} : S \times S \rightarrow S       & \text{merge function}
			\end{array}
		\]
		{\textbf{Network semantics}\quad\fbox{$\st : V \rightarrow (\nat \rightarrow S)$}}
		\begin{align*}
			\st(v)(t) & \in S                               & \text{state at node $v$ at time $t$} \\
			\preds(v) & = \{ u ~\vert~ u \in V, uv \in E \} & \text{in-neighbors of $v$}
		\end{align*}
		\begin{align}
			\st(v)(0)   & = \init{v}       & \label{eq:init-st} \\
			\st(v)(t+1) & = \MF{v}{\st}{t} & \label{eq:next-st}
		\end{align}
	  \end{minipage}
      \Description{A series of equations recapitulating the network model.
      Network instances are presented, along with the network simulation semantics.}
	\caption{Summary of our formal routing model and notation.}
	\label{fig:notation-1}
\end{figure}
\begin{figure}
	  \centering
	\begin{minipage}[t]{\linewidth}
		{\textbf{Interfaces, Properties and Metavariables}}
		\[
			\begin{array}{@{\qquad\ \ }l@{}c@{}l@{\qquad\quad}r}
				A          & \ {} : {} \ {} & V \rightarrow (\nat \rightarrow 2^{S}) & \textrm{node interfaces/invariants} \\
				P          & \ : \          & V \rightarrow (\nat \rightarrow 2^{S}) & \textrm{node properties}          \\
				\predicate & \ : \          & 2^S                                  & \textrm{sets of states}             \\
                \new{\tau} &\ : \           & \nat          & \new{\textrm{witness times}}
			\end{array}
		\]
		{\textbf{Temporal operators} \quad\fbox{$\tpop: \nat \rightarrow 2^{S}$}}
		\[
			\begin{array}{lclr}
				\globally(\predicate)                       & = & \absl{t} \predicate       & \text{globally} \\
				\predicate \until{\tau} \tpop        & = & \absl{t} \mathit{if}\ t < \tau\ \mathit{then} \ \predicate\  \mathit{else}\ \tpop(t) & \text{until}                 \\
				\finally{\tau}(\tpop)                   & = & S \until{\tau} \tpop                                         & \text{finally}               \\
				\tpop_1 \sqcap \tpop_2                             & = & \absl{t} \tpop_1(t) \cap \tpop_2(t)                                                & \text{intersection (lifted)} \\
				\tpop_1 \sqcup \tpop_2                             & = & \absl{t} \tpop_1(t) \cup \tpop_2(t)                                                & \text{union (lifted)}        \\
				\predneg \tpop                                 & = & \absl{t} S \setminus \tpop(t)                                                  & \text{negation (lifted)}
			\end{array}
		\]
		{\textbf{Verification Conditions}}

		{\textit{Initial condition \new{for node $v$}}:}
		\begin{equation}
			\label{eq:init-vc} \init{v} \in \Anno{v}{0}
		\end{equation}

		{\textit{Inductive condition \new{for node $v$ with in-neighbors $u_{1}, u_{2}, \ldots, u_{n}$}}:}
		\begin{equation}
			\label{eq:ind-vc}
			\begin{split}
				&\forall t \in \nat, 
				\forall s_{1} \in \Anno{u_{1}}{t},
				\forall s_{2} \in \Anno{u_{2}}{t},
				\ldots,
				\forall s_{n} \in \Anno{u_{n}}{t}, \\
				& \left(\init{v} \M \MM_{i\in\{1,..,n\}} \T_{u_{i}v}(s_{i}) \right) \in \Anno{v}{t+1}
			\end{split}
		\end{equation}

		{\textit{Safety condition \new{for node $v$}}:}
		\begin{equation}
			\label{eq:safety-vc}
			\forall t \in \nat, \Anno{v}{t} \subseteq \PPr{v}{t}
		\end{equation}
	\end{minipage}
	\caption{Summary of our interfaces and properties, temporal operators and verification conditions.}
	\label{fig:notation-2}
  \Description{
    A series of formulas and notation summarizing three aspects of the formal model.
    At the top, the notation for interfaces, properties and their metavariables is shown.
    We denote interfaces with $\An$, which is a function from nodes $V$ to a function
    from times $\nat$ to sets of states $2^{S}$.
    We denote properties with $\Pr$ and have the same type signature as interfaces.
    We use $\predicate$ to denote sets of states.
    The letter $\tau$ represents a natural number witness time.
    In the middle of the figure are the temporal operators, which are denoted by $\tpop$ and are
    functions from times $\nat$ to sets of states $2^{S}$.
    The globally operator is $\globally(\predicate)$, which is equivalent to a function
    taking a time $t$ as an argument and returning the predicate $\predicate$ (ignoring the argument).
    The until operator is $\predicate \until{\tau} \tpop$, which is equivalent to a function
    taking a time $t$ as an argument and return the predicate $\predicate$ if $t$ is less than $\tau$, and
    otherwise returning the result of calling the temporal operator $\tpop$ with the argument $t$.
    The finally operator is $\finally{\tau}(\tpop)$. This is equivalent to $S \until{\tau} \tpop$,
    meaning a use of the until operator where the predicate is simply ``true'' (the set of all routes $S$).
    Three other operators are given, which lift set operations to temporal operators.
    These are intersection, union and negation.
    All of these take temporal operators as arguments,
    and are equivalent to a function taking a time $t$ and applying the corresponding set operation
    to the results of the temporal operators (when they are called on $t$):
    for instance, intersection $\tpop_{1} \sqcap \tpop_{2}$ of temporal operators $\tpop_{1}$ and $\tpop_{2}$
    is a function taking a time $t$ that returns the set intersection of $\tpop_{1}(t)$ and $\tpop_{2}(t)$.
    The final third of the figure is the verification conditions.
    First is the initial condition for a node $v$, which states that the initial route $\init{v}$ at $v$ must
    be in the interface $\Anno{v}{0}$ at node $v$ at time 0.
    Next is the inductive condition for a node $v$ with in-neighbors $u_{1},u_{2},\ldots,u_{n}$.
    This states that for all times $t$, for all states $s_{1}$ of $u_{1}$'s interface $\Anno{u_{1}}{t}$ at time $t$,
    states $s_{2}$ of $u_{2}$'s interface $\Anno{u_{2}}{t}$ at time $t$, and so on up to $u_{n}$,
    we have a route contained in $v$'s interface $\Anno{v}{t+1}$ at time $t+1$ which
    is the merge of the initial route $\init{v}$ with every transferred route $\T_{u_{i}v}(s_{i})$ for all $u_{i}$
    of the in-neighbors.
    Lastly, the safety condition for a node $v$ states that for all times $t$, the set of states in $v$'s interface
    $\Anno{v}{t}$ at time $t$ is a subset of the set of states in a property $\PPr{v}{t}$.
  }
\end{figure}

Figure~\ref{fig:notation-2} presents our interfaces and language of temporal operators.
As before, we use $\An$ to denote network interfaces.
\new{
  The temporal operators $Q$ are syntactic sugar for \emph{functions}
  from a time $t$ to a predicate.
  $\predicate \until{\tau} Q$ and $\finally{\tau}(Q)$
  take a concrete natural number $\tau$ as a \emph{witness time}
  parameter to compare against the argument time $t$.
  These witness times are \emph{absolute} times,
  specifying the time steps since the initial time 0.
  We can nest $\until{\tau}$ and $\finally{\tau}$ operators
  to represent arbitrarily-many \emph{intervals} of time,
  \EG $\finally{2} (\predicate_{1} \until{4} \globally(\predicate_{2}))$
  is a function that returns $S$ (\TR{}) given input time 0 or 1;
  $\predicate_{1}$ given time 2 or 3;
  and $\predicate_{2}$ given a time of 4 or more.
  We also lift set union, intersection and negation to succinctly
  combine temporal operators.
  In our evaluation (\S\ref{sec:evaluation}),
  we found this (intentionally small) language sufficient to
  express a wide variety of reachability and security properties.}

A valid interface
is an \emph{inductive invariant}~\cite{giannakopoulou2018compositional}.
Such interfaces satisfy the \emph{initial and inductive conditions}
specified in Figure~\ref{fig:notation-2}.
Valid interfaces may be used to prove node properties, as specified
by the \emph{safety condition} in Figure~\ref{fig:notation-2}.
\new{As the network has a finite number of nodes, we can enumerate them
to check these three conditions on every node in the network.}

The most important property of our system is \emph{soundness}: the simulation states
are included in any interface $\An$ that satisfies the initial and inductive
conditions.

\begin{theorem}[Soundness]\label{thm:sound}
	Let $\An$ satisfy the initial and inductive conditions \new{for all nodes}.
	Then $\An$ always includes the simulation state $\st$, meaning
  $\forall v \in V, \forall t \in \nat, \st(v)(t) \in \Anno{v}{t}$.
\end{theorem}
\begin{proof}
	By induction on $t$.
  \iftoggle{appendix}{%
      See~\ref{thm:sound-apx}.
    }{%
      See~\cite{extended} for the full proof.
    }
\end{proof}

Since initial and inductive conditions suffice to prove that simulation states
are included within interfaces, it is safe in turn to use interfaces to
check node properties.

\begin{corollary}[Safety]\label{thm:sound-cor}
	Let $\An$ satisfy the initial and inductive conditions \new{for all nodes}.
  Let $\PP$ satisfy the safety condition with respect to $\An$ \new{for all nodes}.
	Then $\forall v \in V, \forall t \in \nat, \st(v)(t) \in \PPr{v}{t}$.
\end{corollary}
\begin{proof}
	From definitions.
  \iftoggle{appendix}{%
      See~\ref{thm:sound-cor-apx}.
    }{%
      See~\cite{extended}.
    }
\end{proof}

\new{A \emph{closed network} is a network instance where all initial routes are fixed routes,
such that all states of the network $\st$ are captured by a concrete simulation.}
Our verification procedure is \emph{complete} \new{with respect to a closed network}:
for any \new{closed network},
there exists an interface that characterizes its simulation states
exactly.  One of the consequences of completeness is
that our modular verification procedure is powerful enough to prove any
property that we could prove \new{via a concrete simulation}.%
\footnote{\new{In open networks, an external node may take an arbitrary route at
any point in time: this case is not captured by $\st$.}}

\begin{theorem}[Closed Simulation Completeness]\label{thm:complete}
	Let $\st$ be the state of the \new{closed} network.
	Then for all $v \in V$ and $t \in \nat$, $\Anno{v}{t} = \{\st(v)(t)\}$
	satisfies the initial and inductive conditions \new{for all nodes}.
\end{theorem}
\begin{proof}
	By construction of the interface.
  \iftoggle{appendix}{%
      See~\ref{thm:complete-apx}.
    }{%
      See~\cite{extended}.
    }
\end{proof}

\section{SMT Algorithms for Verification}
\label{sec:algorithms}

We can check our initial, inductive and safety verification conditions (VCs) independently
for every node in the network using off-the-shelf SMT solvers.
To check an instance of a VC, the solver will attempt to prove the
condition is valid (\IE true for all choices of $t$ \new{and $s_{1},s_{2}, \ldots, s_{n}$})
\new{by checking whether the \emph{negation} of our original VC is satisfiable.}
If the solver can satisfy the negation, it will provide us with a
\emph{counterexample}---\new{the state(s) of the node(s)} at a particular time such that
the VC does not hold.
Counterexamples to initial or inductive conditions indicate that the interface does not
capture the network's behavior, while a counterexample to the safety condition
indicates that the interface is not strong enough to prove the property.
The latter case may occur because the property is simply not true (indicating a bug),
or because we must strengthen the given interface to prove the property.
\new{If the negation is unsatisfiable, then we know the condition is valid.}

\new{A network instance's routes and behavior determines its encoding in SMT.
  For example, one could encode the route triples in \S\ref{sec:overview} as integer variables
  $\lp$ and $\len$ and a Boolean variable $\bgptag$, and use Presburger arithmetic to encode
  $\TT$ and $\M$ with $+$ and $<$.
  To model networks with external peers (like \fn{n}), multiple routing destinations,
  or other sources of nondeterminism,
  we may use \emph{symbolic variables}.
  For external peers, rather than treating \fn{n} as having a specific initial route such as $\nullm{}$,
  we may use a symbolic variable $s$ for $\init{n}$, and then ask the SMT solver to check the VCs
  \emph{for all} choices of $s$.
  For multiple destinations, we can use a symbolic node variable to choose from a set of destination nodes
  when checking that any node in that set is reachable.
  We also can \emph{assume} arbitrary preconditions for symbolic variables,
  \EG enforcing that $s$ is not tagged with $\neg s.\bgptag$.
  These assumptions are not checked when we encode them to SMT.}

\algrenewcommand\algorithmicfunction{\textbf{proc}}
\paragraph{The modular checking procedure}
We present our modular checking procedure in Algorithm~\ref{alg:check}.
The \textsc{CheckMod} procedure iterates over each node of the network
and encodes the underlying formula (for the current node)
of our three verification conditions
by calling \textsc{EncodeInitCond}, \textsc{EncodeIndCond} and \textsc{EncodeSafeCond}
as defined in \eqref{eq:init-vc}, \eqref{eq:ind-vc} and \eqref{eq:safety-vc}, respectively.%
\footnote{\new{For simplicity,
    our encoding uses predicates $V \rightarrow (\nat \rightarrow (S \rightarrow \bool))$
    to represent $\An$ and $\PP$---other than this change, the formulae are identical to the given VCs.}}
\new{Note that we encode $t$ as an explicit symbolic variable:
  our temporal operators expand to a case analysis over this variable
  to determine what predicate holds on the particular node's route.}
We then
ask the solver if every encoded formula is valid using \valid{}.
If \valid{} returns \FL{} for any check,
we ask for the relevant counterexample using \ctex{},
which returns a variable assignment that violates the formula.
Otherwise, we report success ($\An$ and $\PP$ hold).

\begin{algorithm}[t]
	\caption{\label{alg:check} The modular checking algorithm.}
	\begin{algorithmic}
		\Function{CheckMod}{network $(G,S,\init{},\TT,\M)$, interface $\An$, property $\PP$}\label{alg:check-fn}
		\ForAll{$v \in V$} \textbf{in parallel}
		\State $\formula_{1} \gets \Call{EncodeInitCond}{v, S, \init{}, \An}$
		\Comment{\eqref{eq:init-vc}}
		\State $\formula_{2} \gets \Call{EncodeIndCond}{v, \preds(v), S, \init{}, \TT, \M, \An}$
		\Comment{\eqref{eq:ind-vc}}
		\State $\formula_{3} \gets \Call{EncodeSafeCond}{v, S, \An, \PP}$
		\Comment{\eqref{eq:safety-vc}}
    \For{$i \gets 1, 2, 3$}
      \If{$\neg \Call{\valid{}}{\formula_{i}}$} \Return \Call{\ctex{}}{$\formula_{i}$}
        \EndIf
      \EndFor
		\EndFor
		\State \Return \textsc{Success}
		\EndFunction
	\end{algorithmic}
\end{algorithm}

Encoding the initial and inductive conditions is roughly proportional
in size to the complexity of the policy at the given node,
which in turn is related to the in-degree of the node---denser networks that
include nodes with higher in-degree are more expensive to check.
Encoding the safety condition is proportional to the size of the
formulae describing the interface and property (generally tiny).
In addition to reducing the size of each SMT formula, the
factoring of the problem into independent conditions makes it possible
to check conditions on nodes \emph{simultaneously in parallel}.
We will discuss the performance implications of our procedure further in \S\ref{sec:evaluation}.

\paragraph{Handling counterexamples}
\new{When \ctex{} returns a counterexample for a modular check,
  it specifies
  \begin{enumerate*}[label=(\alph*)]
    \item routes of the node(s);
    \item a concrete time; and
    \item concrete values of any symbolic variables in the formula.
  \end{enumerate*}
These counterexamples can guide us in strengthening the invariants or pinpointing bugs.
This is akin to other invariant-checking tools like Dafny~\cite{dafny} where users must
manually refine their invariants.
Anecdotally, our practical approach to designing the interfaces in \S\ref{sec:evaluation}
was to start by choosing $\Anno{v}{t} = \PPr{v}{t}$ (trivially satisfying our safety condition).
A counterexample would then identify a time instance where $v$'s invariant
violated its initial or inductive condition.
To resolve this violation, we often had to add a $\globally(\predicate)$ invariant to
capture additional behavior (\EG no node receives a better route than the legitimate route).
This ``rule-of-thumb'' suggests that counterexample-guided techniques~\cite{ClarkeCegar00}
may be capable of \emph{inferring invariants}. We leave this as future work.}

\paragraph{Incorporating delay}
\new{$\st$ and $\An$ define a synchronous network semantics.
  As shown in prior work on routing algebras~\cite{metarouting, daggitt2018asynchronous},
  when $\TT$ and $\M$ are \emph{strictly monotonic} --- informally, $\M$ prefers a route $r$
  over any transferred route $\T_{e}(r)$ --- there will only be one converged
  state of the network, which our synchronous model will certainly capture.%
  \footnote{As common protocols (\EG OSPF) rely on shortest-paths algorithms
    with strictly monotonic $\TT$ and $\M$, prior work has sometimes assumed the network converges to a unique solution~\cite{arc, fastplane}.}
  In other cases, the synchronous model captures a possible execution of the network.}

\new{
  We may extend our model to consider routes up to a bounded number of units of delay.
  To account for one unit of delay, we can extend our inductive condition to check
  \emph{all} routes sent in the last \emph{two} time steps $t$ and $t+1$ satisfy the invariant
  at time $t+2$, becoming (changes in boxes):
  \begin{align*}
    \forall t \in \nat,\
	&\forall s_{1} \in \boxed{\Anno{u_{1}}{t} \cup \Anno{u_{1}}{t + 1}},
	\ldots,
	\forall s_{n} \in \boxed{\Anno{u_{n}}{t} \cup \Anno{u_{n}}{t + 1}}, \\
    &\left(\init{v} \M \MM_{i\in\{1,..,n\}} \T_{u_{i}v}(s_{i}) \right) \in \boxed{\Anno{v}{t+2}}
  \end{align*}
  We may extend the condition further to consider more units of delay.
  Doing so may also increase the complexity of our invariants.}

\section{Implementation}
\label{sec:implementation}

We implemented \sysname{}'s modular verification procedure
as a library written in C\#.
The library allows users to construct models of networks
and then modularly verify them.
Like the network modelling framework NV~\cite{nv}, \sysname allows users to customize their models
by providing the type of routes (which may involve integers, strings, booleans,
bitvectors, records, optional data, lists, or sets)
and the initialization, transfer and merge functions that process them.
This modelling language makes it easy to add ghost state to routes, as described earlier.
It is also possible to declare and use symbolic values in the model.
Hence, one may reason about all possible prefixes or more generally about
all possible external routing announcements.

Under the hood, \sysname{} uses Microsoft's Zen verification library~\cite{zen-hotnets}
\new{to generate SMT formulas from higher-level C\# data structures
  and pass these formulas to the \zthree{} SMT solver~\cite{z3}.}
\sysname{} thus supports any network model that Zen can encode to \zthree{}.
For example, to model a network running eBGP, we would adopt many of the modelling
choices made in Minesweeper~\cite{minesweeper}.
We can use integers \new{and Presburger arithmetic} to model path length,
and bitvectors to model local preference and MED.

\sysname{} uses multi-threading to run modular checks in parallel.%
\footnote{We use C\#'s Parallel LINQ library~\cite{plinq},
	which can run up to 512 concurrent threads.}
As each check is independent,
the time to set up additional threads is the only overhead
for parallelization.

\section{Evaluation}
\label{sec:evaluation}

\begin{table}
  \centering
  \caption{Lines of C\# code to \new{define the network instances, interfaces and properties} for each benchmark.\\
    \new{$\dagger$: \bte's network instance is defined partly from Internet2's configuration files
      (topology $G$ and $\TT$ functions); the remaining elements ($\init{}$, $\M$ and symbolic variables) are defined in C\#.
    The reported lines are for the C\# portion---the configuration files are over 100,000 lines of Junos configuration code.}}
  \label{fig:interface-loc}
  \begin{tabular}{lccc}
	\toprule
	\textbf{Benchmark} & \new{\textbf{Network LoC}} & \textbf{Interface LoC} & \textbf{Property LoC} \\ \midrule
	\reach & 79 & 3 & 2 \\
	\length & 83 & 7 & 5 \\
	\vf & 87 & 12 & 2 \\
	\hijack & 142 & 21 & 4 \\
	\bte & 83$\dagger$ & 5 & 5 \\
	\bottomrule
  \end{tabular}
\end{table}

To evaluate \sysname{} and illustrate its scaling trends, we generated
a series of synthetic fattree~\cite{fattree} data center networks
and verified four variations on reachability properties.
We also verified an isolation property on a real wide-area network configuration
with over 100,000 lines of \new{Junos configuration} code.
\new{These two types of networks demonstrate \sysname{}'s performance for
  networks which are highly-connected (data centers) and have complex policies
  (WANs).}
Table~\ref{fig:interface-loc} shows that writing the interfaces for each of our
benchmarks is low-effort compared to the rest of the network in terms of lines of code.
We generated interfaces for our experiments parametrically for any size of network,
based on the distinct \emph{roles} of nodes:
for fattree networks, a node's pod and tier determined its role (5 roles, discussed later);
for our wide-area network benchmark, we distinguished internal nodes from external neighbors (2 roles).
\new{As a node's role determined its invariant, it is easy to \emph{extend or update} these networks
  (\EG adding a new pod or external neighbor) and reuse the appropriate existing invariant for
a node of that role.}

To compare our implementation against a baseline,
we implemented a monolithic, network-wide Minesweeper-style~\cite{minesweeper}
procedure \ms{}
and compared its performance against \sysname{}.
\ms{} analyzes stable states, which are independent of time.
\new{Given a property $\PP$ over stable states, \ms{} checks if $\PP$
  always holds (is valid) given the stable states of the network.
  These states are encoded as a single formula over all nodes,
  as described in \S\ref{sec:background}.}
To compare \ms{} with \sysname{}, we first crafted properties
for \sysname{}, which employs timed invariants.
We then erased the temporal details from these invariants
to generate properties that \ms{} could manage.
For instance, when \sysname{} would verify properties of
the form $\globally(\predicate)$, $\finally{t}\globally(\predicate)$, or
$\predicate_2 \until{t}\globally(\predicate)$, \ms{} would instead verify
that the network's stable states satisfy $\predicate$.%
\footnote{\new{None of our properties required that we specify more than one witness time,
    so all $\finally{}$ and $\until{}$ operators took a $\globally$ operator as their temporal argument.}}

We ran our benchmarks on a Microsoft Azure D96s v5 virtual machine with 96 vCPUs and 384GB of RAM.
We used the machine's multi-core processor to run all modular checks in parallel,
while monolithic checks necessarily ran on a single thread.
We timed out any benchmark that did not complete in 2 hours.
We report four times for each benchmark:
\begin{enumerate*}[label=\emph{(\roman*)}]
	\item the total time until all \sysname{} threads finished (\sysshort);
	\item the median node check time;
	\item the $99^{\text{th}}$ percentile node check time (99\% of checks completed in less than this much time); and
	\item the total time taken by \ms.
\end{enumerate*}

\paragraph{Fattrees}
We parameterize our fattree networks by the number of pods $k$: a $k$-fattree has
$1.25 k^{2}$ nodes and $k^{3}$ edges:
Figure~\ref{fig:fat} shows an example 4-fattree used in our $\vf$ benchmark.
We considered multiples of $4$ for $4 \leq k \leq 40$ to assess \sysname{}'s
scalability: whereas we expected a monolithic
verifier to time out on larger topologies, we hypothesized that
\sysname{} would scale to these networks.
We present how verification time grows with respect to the number of nodes
in each fattree in Figure~\ref{fig:smt-fat}.
The figure shows verification time \new{on the y-axis} on a logarithmic scale.

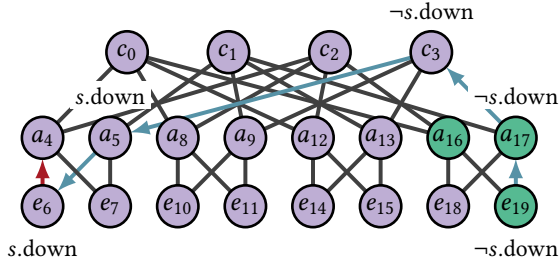
\begin{figure}
	\centering
	\begin{tikzpicture}
	  \SetDistanceScale{0.9}
		\foreach \x [count=\xname from 0] in {-.75,0.75,2.25,3.75} {
				\Vertex[x=\x,y=2.25,label=$c_{\xname}$,size=.55]{C\xname}
			}
		\foreach \xname [count=\x from -2] in {4,5,8,9,12,13} {
				\Vertex[x=\x,y=1,label=$a_{\xname}$,size=.55]{A\xname}
			}
		\Vertex[x=4,y=1,label=$a_{16}$,size=.55,color=nodegreen]{A16}
		\Vertex[x=5,y=1,label=$a_{17}$,size=.55,color=nodegreen]{A17}
		\foreach \xname [count=\x from -2] in {6,7,10,11,14,15,18} {
				\Vertex[x=\x,y=0,label=$e_{\xname}$,size=.55]{E\xname}
			}
		\Vertex[x=5,y=0,label=$e_{19}$,size=.55,color=nodegreen]{E19}
		\foreach \agg in {4,8,12,16} {
				\Edge(C0)(A\agg)
				\Edge(C2)(A\agg)
			}
		\foreach \agg in {9,13} {
				\Edge(C1)(A\agg)
				\Edge(C3)(A\agg)
			}
		\Edge(A5)(C1)
		\Edge(A17)(C1)
		\Edge[Direct,RGB,color={86,146,166}](A17)(C3)
		\Edge[Direct,RGB,color={86,146,166}](C3)(A5)
		\Edge[Direct,color=darkred](E6)(A4)
		\Edge[Direct,RGB,color={86,146,166}](A5)(E6)
		\Edge(E7)(A4)
		\Edge(E7)(A5)
		\Edge(E10)(A8)
		\Edge(E10)(A9)
		\Edge(E11)(A8)
		\Edge(E11)(A9)
		\Edge(E14)(A12)
		\Edge(E14)(A13)
		\Edge(E15)(A12)
		\Edge(E15)(A13)
		\Edge(E18)(A16)
		\Edge(E18)(A17)
		\Edge(E19)(A16)
		\Edge[Direct,RGB,color={86,146,166}](E19)(A17)
		\Text[x=5,y=-.6,fontsize=\small,style={fill=white!8}]{$\neg s.\downtag$}
		\Text[x=5,y=1.62,fontsize=\small,style={fill=white!8}]{$\neg s.\downtag$}
		\Text[x=3.75,y=2.85,fontsize=\small,style={fill=white!8}]{$\neg s.\downtag$}
		\Text[x=-1,y=1.62,fontsize=\small,style={fill=white!8}]{$s.\downtag$}
		\Text[x=-2,y=-.6,fontsize=\small,style={fill=white!8}]{$s.\downtag$}
	\end{tikzpicture}
    \Description{
      A graph representing a 4-pod fattree network with 20 nodes.
      The 4 core (\AKA spine) nodes are labelled $c_{0}, c_{1}, c_{2}$ and $c_{3}$.
      The 8 aggregation nodes are labelled $a_{4}, a_{5}, a_{8}, a_{9}, a_{12}, a_{13}, a_{16}$ and $a_{17}$.
      Finally, the 8 edge (\AKA top-of-rack or ToR) nodes are labelled $e_{6}, e_{7}, e_{10}, e_{11}, e_{14}, e_{15}, e_{18}, e_{19}$.
      There is a route indicated by a blue arrow as departing node $e_{19}$ with the $s.\downtag$ field set to false,
      and which travels up to $a_{17}$ and then $c_{3}$ with the tag still set to false.
      From $c_{3}$, it travels down to $a_{5}$ in a different fattree pod, where the tag is now set to true.
      It then travels from $a_{5}$ down to $e_{6}$.
      From here, a red arrow indicates the route at $e_{6}$ will not be shared up the edge to $a_{4}$, since the tag is set to true.
    }
	\caption{An example fattree network, showing
		how $\vf$ sets $s.\downtag$ along the path between
		the destination node $e_{19}$ and $e_{6}$.
		$e_{6}a_{4}$ will drop the route from $e_{6}$ to prevent valley routing.
	}
	\label{fig:fat}
\end{figure}

\begin{figure*}
	\begin{tikzpicture}
		\begin{groupplot}[
				group style={
						group size=4 by 2,
						horizontal sep=2mm,
						vertical sep=1.1cm,
						xlabels at=edge bottom,
						ylabels at=edge left,
						yticklabels at=edge left,
					},
				enlargelimits=true,
				height=3.6cm,
				width=0.274\linewidth,
				grid=major,
				ymin=0.5,
				ymax=10000,
				ytickten={0,...,4},
				xmin=0,
				xmax=2000,
				xlabel=Nodes,
				ymode=log,
				ylabel=Verification time,
				y unit=\si{\second},
				title style={
						text width=0.3\linewidth,
					},
				cycle list name=cute
			]
			\nextgroupplot[thick,title={\subcaption{\label{fig:spreach}\spreach}},title style={yshift=-0.8em}]
			\foreach \col in {tk,med,99,ms} {
					\addplot table[x expr=\thisrow{k}^2*1.2, y expr=\thisrow{\col}/1000]{\rtbl};
				}
			\addplot[timeout line] {7200.0} node [timeout node] {timeout};
			\nextgroupplot[thick,title={\subcaption{\label{fig:splength}\splength}},title style={yshift=-0.8em}]
			\foreach \col in {tk,med,99,ms} {
					\addplot table[x expr=\thisrow{k}^2*1.2, y expr=\thisrow{\col}/1000]{\ltbl};
				}
			\addplot[timeout line] {7200.0} node [timeout node] {timeout};
			\nextgroupplot[thick,title={\subcaption{\label{fig:vfnet}\vfnet}},title style={yshift=-0.8em}]
			\foreach \col in {tk,med,99,ms} {
					\addplot table[x expr=\thisrow{k}^2*1.2, y expr=\thisrow{\col}/1000]{\vtbl};
				}
			\addplot[timeout line] {7200.0} node [timeout node] {timeout};
			\nextgroupplot[thick,title={\subcaption{\label{fig:hijnet}\hijnet}},title style={yshift=-0.8em},legend style={at={(2.2,1.0)}}]
			\addplot table[x expr=\thisrow{k}^2*1.2, y expr=\thisrow{tk}/1000]{\htbl};
			\addlegendentry{\sysshort}
			\addplot table[x expr=\thisrow{k}^2*1.2, y expr=\thisrow{med}/1000]{\htbl};
			\addlegendentry{\sysshort median}
			\addplot table[x expr=\thisrow{k}^2*1.2, y expr=\thisrow{99}/1000]{\htbl};
			\addlegendentry{\sysshort $99^{th}$ p.}
			\addplot table[x expr=\thisrow{k}^2*1.2, y expr=\thisrow{ms}/1000]{\htbl};
			\addlegendentry{\ms}
			\addplot[timeout line] {7200.0} node [timeout node] {timeout};
			\nextgroupplot[thick,title={\subcaption{\label{fig:apreach}\apreach}},title style={yshift=-0.8em}]
			\foreach \col in {tk,med,99,ms} {
					\addplot table[x expr=\thisrow{k}^2*1.2, y expr=\thisrow{\col}/1000]{\artbl};
				}
			\addplot[timeout line] {7200.0} node [timeout node] {timeout};
			\nextgroupplot[thick,title={\subcaption{\label{fig:aplength}\aplength}},title style={yshift=-0.8em}]
			\foreach \col in {tk,med,99,ms} {
					\addplot table[x expr=\thisrow{k}^2*1.2, y expr=\thisrow{\col}/1000]{\altbl};
				}
			\addplot[timeout line] {7200.0} node [timeout node] {timeout};
			\nextgroupplot[thick,title={\subcaption{\label{fig:apvfnet}\apvfnet}},title style={yshift=-0.8em}]
			\foreach \col in {tk,med,99,ms} {
					\addplot table[x expr=\thisrow{k}^2*1.2, y expr=\thisrow{\col}/1000]{\avtbl};
				}
			\addplot[timeout line] {7200.0} node [timeout node] {timeout};
			\nextgroupplot[thick,title={\subcaption{\label{fig:aphijnet}\aphijnet}},title style={yshift=-0.8em}]
			\addplot table[x expr=\thisrow{k}^2*1.2, y expr=\thisrow{tk}/1000]{\ahtbl};
			\addplot table[x expr=\thisrow{k}^2*1.2, y expr=\thisrow{med}/1000]{\ahtbl};
			\addplot table[x expr=\thisrow{k}^2*1.2, y expr=\thisrow{99}/1000]{\ahtbl};
			\addplot table[x expr=\thisrow{k}^2*1.2, y expr=\thisrow{ms}/1000]{\ahtbl};
			\addplot[timeout line] {7200.0} node [timeout node] {timeout};
		\end{groupplot}
	\end{tikzpicture}
    \Description{
      8 graphs plotting verification times for \sysname{} compared to \ms{}.
      The graphs show the eight fattree benchmarks ---
      the single-destination benchmarks \spreach, \splength, \vfnet, \hijnet and
      the all-pairs edge destination benchmarks \apreach, \aplength, \apvfnet and \aphijnet.
      All the graphs show the number of nodes on the x-axis from 0 to 2000,
      and the y-axis shows verification times on a logarithmic scale from 1 second to
      10000 seconds.
      A timeout line is drawn at 7200 seconds for the 2-hour timeout.
      The graphs show the four measured times: the total time of \sysname{} in blue,
      its median check time in orange, its 99th percentile check time in green,
      and the total time of \ms{} in red.
      For all benchmarks except \spreach, \sysname{} outperforms \ms{}: in 5 of the 8 benchmarks,
      \ms{} times out at the 80 node data point, whereas \sysname{} scales to 2000 nodes in all benchmarks
      except the \apvfnet benchmark.
    }
	\caption{\ms \VS \sysshort verification times for fattree benchmarks with 8 different policies.}
	\label{fig:smt-fat}
\end{figure*}
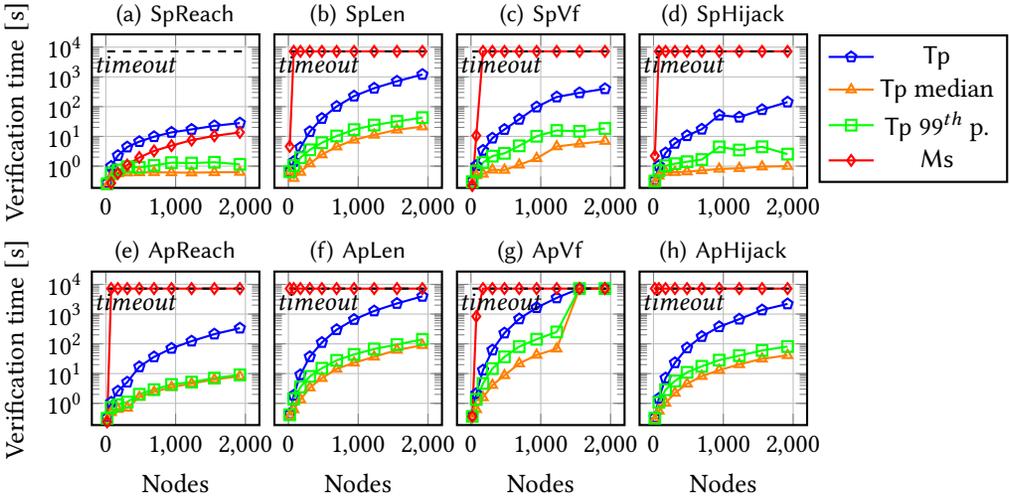

We considered four different properties \new{(explained below)}:
reachability (\reach),
bounded path length (\length), valley freedom (\vf) and route filtering (\hijack).
We considered each property when routing to a single, fixed destination edge node $\dest$ (\singleprefix),
and routing to a $\dest$ edge node \new{determined by a symbolic variable (\allprefix) ---
modelling all-pairs routing to any edge node.}
Our routes modelled the eBGP protocol in these networks.
Table~\ref{fig:route-fields} summarizes the eBGP fields represented
and how we modelled them in SMT.
We model the major common elements of eBGP routing:
a route destination as a 32-bit integer
(representing an IPv4 prefix);
administrative distance, local preference, multi-exit discriminators
as 32-bit integers (encoded as bitvectors);
eBGP origin type as a ternary value;
the AS path length as an (unbounded) integer; and
BGP communities as a set of strings.

\begin{table}
	\centering
	\caption{eBGP route fields modelled by \sysname{} in SMT for fattree benchmarks.}
	\label{fig:route-fields}
	\begin{tabular}{ lc }
		\toprule
		\textbf{Route field}          & \textbf{Modelled type in SMT}                                         \\ \midrule
		Route destination             & bitvector~\cite{smt-bitvector}                                          \\
		Administrative distance       & bitvector~\cite{smt-bitvector}                                          \\
		eBGP local preference         & bitvector~\cite{smt-bitvector}                                          \\
		eBGP multi-exit discriminator & bitvector~\cite{smt-bitvector}                                          \\
		eBGP origin type              & $\{\mathrm{egp}, \mathrm{igp}, \mathrm{unknown}\}$~\cite{smt-bitvector} \\
		eBGP AS path length           & integer~\cite{smt-integer}                                              \\
		eBGP communities              & set<string> ~\cite{smt-string, smt-array}                               \\
		\bottomrule
	\end{tabular}
\end{table}

\paragraph{Witness times}
Our temporal operators include witness times to specify the time when a node $v$ has a route.
To prove that $v$ \emph{eventually} has a route to the destination node $\dest$ at time $t$,
we must show that one of $v$'s neighbors sent it such a route at an earlier time $t - 1$.
The exact time $v$ obtains a route thus depends on where $\dest$ is relative to $v$.
This breaks down to five cases (or \emph{roles}),
following the fattree's structure, depending on whether $v$ is:
\begin{enumerate*}[label=\emph{(\roman*)}]
  \item the $dest$ node (0 hops, has a route at $t = 0$);
  \item an aggregation node in $\dest$'s pod (1 hop, has a route at $t = 1$);
  \item a core node, or an edge node in $\dest$'s pod (2 hops, $t = 2$);
  \item an aggregation node in another pod ($t = 3$); or
  \item an edge node in another pod ($t = 4$).
\end{enumerate*}
\new{The largest choice of $t$ is hence 4, the diameter of the fattree.}
These cases mirror those identified for local data center invariants in~\cite{secguru}.
We use a function $\distance(v)$ to encode these cases for each node $v$.

\paragraph{\reach}
\reach demonstrates the simplest possible routing behavior and serves as a useful baseline.
The policy simply increments the path length of a route on transfer.
We initialized one destination edge node with a route to itself,
and all other nodes with no route ($\nullm{}$).
Our goal is to prove every node eventually has a route to the destination (\IE its route is not $\nullm{}$).
More precisely, because our network has diameter 4, each node $v$ should
acquire a route in 4 time steps.
\begin{equation*}
	\PP_{\reach}(v) \mydef \finally{4}\globally(s \neq \nullm{})
\end{equation*}

Interfaces for these benchmarks mirror the simplicity of the policy and property.
If a node's route $s \neq \nullm{}$ at time $t$, then its neighbors will in turn have
a route $s \neq \nullm{}$ at time $t+1$.
\begin{equation*}
	\An_{\reach}(v) \mydef \finally{\distance(v)}\globally(s \neq \nullm{})
\end{equation*}

\spreach's policy and property are so simple that \sysshort is actually slightly \emph{slower}
than \ms, as shown in Figure~\ref{fig:spreach}.
We conjecture the \ms{} encoding reduces to a particularly easy SAT instance.
That said, we can already see that individual checks in \sysname{}
take only a fraction of the time that \ms takes,
with 99\% of node checks completing in at most 1.1 seconds, even for
our largest benchmarks.
For \apreach,
Figure~\ref{fig:apreach} shows that, perhaps from the burden of modelling the symbolic
$\dest$, \ms times out at $k=8$.
\sysshort verifies our largest benchmark ($k=40$, with 2,000 nodes) in 5.5 minutes, with 99\% of individual node
checks taking under 9 seconds.

\paragraph{\length}
Our next benchmark uses the same policy as \reach, but considers a stronger property:
every node eventually has a route of at most 4 hops to the destination.
\begin{equation*}
	\PP_{\length}(v) \mydef \finally{4}
	\globally\big(s.\len \leq 4\big)
\end{equation*}

To prove this property, our interfaces specify that path lengths in routes should not
exceed the distance to the destination: $s.\len \leq \distance(v)$.
In addition, because local preference influences routing,
we fix the local preference to the default for all routes when present:
$s.\lp = 100$.
\begin{equation*}
  \An_{\length}(v) \mydef
  	\underbrace{\globally \big( s = \nullm{} \vee s.\mathrm{lp} = 100 \big)}_{\textit{no better routes appear}}
	\sqcap \underbrace{\finally{\distance(v)} \globally\big(s.\len \leq \distance(v) \big)}_{\textit{eventually the route appears} }
\end{equation*}

Reasoning over path lengths requires Z3 to use slower bitvector and integer theories.
Figure~\ref{fig:splength} shows that monolithic verification times out at $k=12$ for \splength.
By contrast, modular verification is able to solve $k=40$
in just over 20 minutes, with 99\% of nodes verified in under 43 seconds.
Figure~\ref{fig:aplength} shows that monolithic verification is not even possible for \aplength at $k=4$;
\sysshort{} completes for \aplength $k=40$ in around 66 minutes, with 99\% of nodes verified in 2.4 minutes.

\paragraph{\vf}
\vf extends \reach with policy to prevent up-down-up (valley)
routing~\cite{valleyfreedc,propane,propane-at}, where routes transit an intermediate pod.
To implement this policy, we add a BGP community $D$
along ``down'' edges in the topology
(\IE from a core node or from an aggregation node to an edge node), and drop
routes with $D$ on ``up'' edges
(see Figure~\ref{fig:fat}).
For brevity, we write ``$s.\downtag$'' to mean ``$D \in s.\tags$''.
We test the same reachability property as \reach.

The legitimate routes in the fattree all start as routes travelling up
from the destination node's pod, \EG the nodes in green ($a_{16},a_{17},e_{19}$) in Figure~\ref{fig:fat}.
We refer to these nodes as ``adjacent nodes''
with a shorthand $\adjacent(v)$:
they transmit routes to the core nodes (and thereby to the rest of the network)
along their up edges.
These edges will drop the routes if $s.\downtag$, so we require
that $\adjacent(v) \rightarrow \neg s.\downtag$.
To ensure this, we add conjuncts to our interfaces
requiring that nodes' final routes are no better than the shortest path's route:
$s.\lp = 100 \wedge s.\len = \distance(v)$.
This ensures our inductive condition holds after every node has a route:
otherwise, a core node (for instance) could offer a spurious route with
$s.\len < 1 \wedge s.\downtag$ to an adjacent node.
\begin{equation*}
  \An_{\vf}(v) \mydef s = \nullm{} \until{\distance(v)}
  	\globally\Big( \underbrace{s.\lp {=} 100 \wedge s.\len {=} \distance(v)}_{\textit{no better routes appear}} \wedge
	\underbrace{\big( \adjacent(v) \rightarrow \neg s.\downtag \big)}_{\textit{adjacent nodes will share routes}} \Big)
\end{equation*}

Figure~\ref{fig:vfnet} shows that \ms verifies up to $k=8$ before timing out.
As with \splength, \sysshort time grows gradually
in proportion to the number of nodes, topping out at 6.6 minutes for $k=40$,
with all node checks completing in under 20 seconds.
Figure~\ref{fig:apvfnet} shows that for all-pairs routing,
monolithic verification times out again at $k=12$, whereas
\sysshort hits the 2-hour timeout at $k=36$.
We conjecture this may be due to the added complexity of encoding $\adjacent(v)$
when the destination is symbolic.

\paragraph{\hijack}
\hijack models a fattree with an additional ``hijacker'' node $h$ connected
to the core nodes.
$h$ represents a connection to the Internet from outside the network,
which may advertise \emph{any} route.
We add a boolean ghost state $\bgptag$ to $S$ for this policy
to mark routes as external (from $h$) or internal.
The destination node will advertise a route with $s.\prefix = p$,
where $p$ is a symbolic value representing an internal address:
the core nodes will then drop any routes from $h$ for prefix $p$,
but allow other routes through.
Apart from this filtering, routing functions as in the \reach benchmarks.
For this network, we verified that every internal node eventually has a route
for prefix $p$ and which is not via the hijacker ($\neg s.\bgptag$),
assuming nothing about the hijacker's route ($\An_{\hijack}(h) \mydef \globally(\TR)$).
%
\begin{equation*}
	\PP_{\hijack}(v) \mydef \finally{4}\globally(s.\prefix {=} p \wedge \neg s.\bgptag)
\end{equation*}

The \hijack interface is straightforward.
We must simply re-affirm that nodes with internal prefixes never have external routes:
$s.\prefix = p \rightarrow \neg s.\bgptag$.
Once nodes have received a route from the destination at time $\distance(v)$,
they should keep that route forever, and hence their route
will have both $s.\prefix = p$ and $\neg s.\bgptag$.
\begin{equation*}
  \An_{\hijack}(v) \mydef
  	\underbrace{\finally{\distance(v)}\globally(s.\prefix {=} p \wedge \neg s.\bgptag)}_{\textit{route will be internally reachable}}
  	\sqcap \underbrace{\globally(s.\prefix {=} p \rightarrow \neg s.\bgptag)}_{\textit{no hijack route is ever used}}
\end{equation*}

In \hijnet, monolithic verification times out at $k=8$, whereas
modular verification time scales to $k=40$.
99\% of nodes complete their checks in under 3 seconds,
with our longest check taking 10.7 seconds at $k=40$.
As with our other benchmarks, verification time grows linearly with respect to the in-degree of each node
(which determines the size of the SMT encoding of our inductive condition).
Figure~\ref{fig:aphijnet} shows similar patterns for the all-pairs case,
  with \emph{no} monolithic benchmark completing on time, and modular verification
taking at most 36.6 minutes.

\paragraph{Wide-area networks}
To better investigate \sysname{}'s scalability for other types of networks,
we evaluated it on the Internet2~\cite{internet2} wide-area network.%
  \footnote{We used the versions of Internet2's configuration files available here~\cite{internet2-configs}.}
\new{Internet2's configuration files are one of the few publicly available examples
  of the complex policies of wide-area and cloud provider networks,
  as described in Propane~\cite{propane}.
  These files contain 1,552 Junos routing policies,
  which filter routes by tag or prefix and tag routes according to customer priorities
  (\EG commercial \VS academic peers).}
We converted the configuration files
to \sysname{}'s model by extracting the policy details using Batfish~\cite{batfish}.
The resulting network has \new{over 200 nodes}: 10 internal nodes within Internet2's AS and 253 external peers.
We did not model all components of Internet2's routing policies:
we focused on IPv4 and BGP routing,
and treated some complex behaviors as ``havoc'' (soundly overapproximating the true behavior).%
\footnote{These include prefix matching, community regex matching and AS path matching.
We also did not model BGP nexthop.}
We do not know Internet2's intended routing behavior:
because of this, we cannot be certain that a counterexample found by \sysname{} represents a real
violation of the network's behavior;
nonetheless, we may still use this network to assess how well \sysname{} enables modular verification.


It appears that Internet2 uses a \textbf{BTE} community tag to identify routes that
must not be shared with external neighbors.
We checked that, if the internal nodes initially have any possible route, then
no external neighbor of Internet2 should ever obtain a route with the \textbf{BTE} tag
set, assuming the external neighbors do not start with such routes.

\begin{equation*}
  \PP_{\bte}(v) \mydef
  \begin{cases}
	\globally \big( s \neq \nullm \rightarrow \mathbf{BTE} \notin s.\tags \big) &\text{if $v$ is external} \\
	\globally (\TR) & \text{otherwise}
  \end{cases}
\end{equation*}

Our interface is the property, \IE $\forall v.~\An_{\bte}(v) \mydef \PP_{\bte}(v)$.
Modular checking remains fast despite the network's more complex policies:
on a 6-core Macbook Pro with 16GB of RAM, modular verification completes in 38.3 seconds, with
a median check time of 0.6 seconds and a $99^{\text{th}}$ percentile check time of 4.2 seconds.
Monolithic verification does not complete after 2 hours.

\section{Related Work}
\label{sec:related-work}

Our work is most closely related to other efforts in
control plane verification~\cite{batfish,arc,minesweeper,tiramisu,fastplane,shapeshifter,plankton,hoyan,nv,lightyear,kirigami,bagpipe,bonsai}.
\new{We separate these tools into classes
  based on the specifics of one's verification problem.}

\paragraph{SMT-based verification}
\new{For networks which are small (in the tens of nodes),
SMT-based tools such as Minesweeper~\cite{minesweeper} or Bagpipe~\cite{bagpipe}
offer ease-of-use, generality, and symbolic reasoning.}
Minesweeper supports a broad range of properties including reachability, waypointing, no blackholes and loops, and device equivalence.

\paragraph{Simulation-based verification}
\new{For larger networks that do not necessitate incremental recomputation after device update
nor fully symbolic reasoning, one can use simulation-based tools~\cite{batfish,fastplane,shapeshifter,plankton,hoyan,nv}.}
Some of these tools also employ symbolic reasoning in limited ways to provide useful capabilities.
For example, inspired by effective work on data plane analysis~\cite{veriflow},
Plankton~\cite{plankton} first analyzes configurations to identify IP prefix equivalence classes.
Identified equivalence classes may be treated symbolically in the rest of the computation.
Plankton might be able to reason symbolically about other attributes,
but doing so would require additional custom engineering
to find the appropriate sort of equivalence class ahead of time (\EG for BGP AS paths or communities).
\sysname{} \new{can represent any or all route fields symbolically (\EG for external actors),
  or determine the behavior of $\init{}$, $\TT$ and $\M$ functions using symbolic variables (\EG for all-pairs properties).}
The solver effort is then passed off to the underlying SMT engine.

\paragraph{Scalable SMT-based verification}
\new{Fewer tools exist for networks which are large and where symbolic reasoning is important.}
Bonsai exploits symmetry to derive smaller abstractions of a network~\cite{bonsai},
but does not work if the network has topology or policy asymmetries or one considers
failures (which break topological symmetries).
Other tools exploit modularity in network designs.
Kirigami~\cite{kirigami} verifies networks using assume-guarantee reasoning, 
but requires interfaces to specify the exact routes passed between any two components.
As such, it is impossible to craft interfaces that are robust to minor changes in network policies.
\new{Lightyear~\cite{lightyear} allows users to craft more general interfaces, but can only prove
properties that capture the \emph{absence} of some ``bad'' route (\EG our \bte{} property).
Lightyear does not model route interaction (\EG selecting routes by path length),
which we do via our $\M$ function: Lightyear's verification queries are therefore smaller,
as it can check invariants on every edge independently.}
We conjecture (but have not proven) that Lightyear verifies properties that \sysname{} expresses as
$\globally (\predicate)$, but not properties requiring $\until{t}$ or $\finally{t}$ temporal operators.

\paragraph{Other efforts in network analysis}
Daggitt \ETAL also use a timed model~\cite{daggitt2018asynchronous},
but focus on convergence properties of routing protocols.
We analyze properties that depend upon a network's topology and configuration such as reachability.

Other inspirations for our work are SecGuru and RCDC~\cite{secguru} in
data plane verification.
Unlike our work, they use non-temporal invariants, which
they extract from the network topology
and assume as ground truth for policies on individual devices.
Whereas our experiments likewise used the topology to define local invariants,
our verification procedure checks that these invariants are in fact guaranteed by the
other devices in the network.

\paragraph{Compositional reasoning}
Our work is inspired by the success of automated methods for compositional verification of concurrent systems -- a recent handbook chapter~\cite{giannakopoulou2018compositional} provides many useful pointers to the rich literature on this topic.
Automated methods using compositional reasoning have been successfully applied in many application domains -- concurrent programs (\EG~\cite{OwickiG76,flanagan2003thread,threaderCav11}), hardware designs (\EG~\cite{Kurshan88,McMillan97,HenzIccad00}, reactive systems (\EG~\cite{alur1999reactive}) --- and for a range of properties including safety and liveness, as well as for refinement checking.
\cite{lomuscio2010assume} applies assume-guarantee reasoning for verifying stability of network congestion control systems.
Many such applications use temporal logic for specifying properties, as well as assumptions and guarantees at component interfaces~\cite{pnueli-modular}.
One main challenge is to come up with suitable assumptions that are strong enough to prove the properties of interest.
Toward this goal, \sysname uses a language of temporal invariants 
inspired by temporal logic to support checking 
local (\IE per-router) properties. 
However, it carefully limits the expressiveness of this language, \EG by not allowing \new{arbitrary} nesting of temporal operators, while allowing efficient verification of the proof obligations using SMT solvers.

Another main difference is that unlike most existing methods, \sysname uses time as an \emph{explicit} variable $t$ in the language of invariants.
This serves two distinct but related purposes.
First, $t$ provides a well-founded ordering to ensure that our proof rule is sound.
Second, using $t$ explicitly in the language of invariants avoids choosing some static ordering over the components, which could be otherwise used to break a circular chain of dependencies between their assumptions
(\CF the \textsc{Circ} rule in~\cite{giannakopoulou2018compositional}).
Unfortunately, it is not always possible to determine a static ordering,
especially in cases where we consider multiple destinations at once symbolically.

Others have used induction over time~\cite{Misra81} or over traces in specific models such as compositions of Moore/Mealy machines~\cite{McMillan97,HenzToplas02} and reactive modules~\cite{alur1999reactive,HenzIccad00} to prove soundness of circular assume-guarantee proof rules.
However, to the best of our knowledge, no prior efforts use time explicitly in the language of assumptions.
Although handling time explicitly could be more costly for decision procedures,
in practice we use abstractions (via temporal operators) that result in fairly compact formulas.
Our evaluations show that these formulas can be handled well by modern SMT solvers.

There have also been many efforts that \emph{automatically} derive assumptions for compositional reasoning~\cite{giannakopoulou2018compositional}.
Representative techniques include computing fixed points over localized assertions called \emph{split} invariants~\cite{NamjoshiCav07}, learning-based methods~\cite{Cobleigh03}, and counterexample-guided abstraction refinement~\cite{PasareanuCav08,GrumbergFac18}.
We can view our interfaces as split invariants, since they refer to only the local state of a component.
However, we depend on the users to provide them as annotations.

\paragraph{Modular verification of distributed systems}
There have been many prior efforts~\cite{iris,TatlockPopl18,ironfleet,I4,DistAI,ivy,modp} for modular verification of distributed systems -- see a recent work~\cite{TatlockPopl18} for other useful pointers.
In general, they handle much richer program logics or computational models than the network routing algebras we target; hence the required assumptions
and verification tasks are more complex, and often require interactive theorem-proving.
The synchronous semantics of network routing algebras~\cite{daggitt2018asynchronous} that
underlies our work is more closely related
to hardware designs modelled as compositions of finite state machines (FSMs),
where a component FSM's state at time $t+1$ depends on its state at time $t$ and new inputs at time $t+1$,
some of which could be outputs from other FSM components, \IE their state at time $t$.
No existing efforts for such models (\EG~\cite{McMillan97,HenzToplas02}) consider time explicitly in the assumptions.

\section{Conclusion}
\label{sec:conclusion}

Ensuring correct routing is critical to the operation of reliable networks.
To verify today's hyperscale networks,
we need modular control plane verification techniques that are general, expressive and efficient.
We propose \sysname{}, a radical new approach for verifying control planes based on
a temporal foundation, which splits the network into small modules to verify efficiently in parallel.
To carry out verification, users provide \sysname{} with local interfaces using temporal operators.
We proved that \sysname{} is sound with respect to the network semantics
\new{and complete for closed networks},
and argue that its temporal foundation is an excellent choice for modular verification.

\begin{acks}
We thank our anonymous reviewers and our shepherd, Jedidiah McClurg, for their helpful feedback.
This work was supported in part by grants from the \grantsponsor{npi}{Network Programming Initiative}{https://www.network-programming.org/}
and the \grantsponsor{nsf}{National Science Foundation}{https://www.nsf.org/}:
\grantnum{nsf}{1837030},
\grantnum{nsf}{2107138}.
\end{acks}

\section*{Availability}
\sysname{} is publicly available via GitHub~\cite{github-repo} as well as Zenodo~\cite{zenodo-repo}.
These contain the necessary code to run the benchmarks in our evaluation or write new ones in C\#.
We include a Docker image and Makefile to assist in running the benchmarks.

\bibliographystyle{ACM-Reference-Format}
\bibliography{references,ref-compo}

\iftoggle{appendix}{%
  \appendix
  \section{Proofs}
  \label{sec:proofs}

  We present the full proofs of Theorem~\ref{thm:sound}, Corollary~\ref{thm:sound-cor}
  and Theorem~\ref{thm:complete} below.
  \new{For clarity, we use boxes, \EG $\boxed{x}$, to indicate where we substitute $x$ for another term.}

  \begin{theorem}[Soundness]\label{thm:sound-apx}
    Let $\An$ satisfy the initial and inductive conditions \new{for all nodes}.
    Then $\An$ always includes the simulation state $\st$, meaning
    $\forall v \in V, \forall t \in \nat, \st(v)(t) \in \Anno{v}{t}$.
  \end{theorem}
  \begin{proof}
    Let $v$ be an arbitrary node in $V$.
    The proof is straightforward by induction over time.

    At time 0, we have $\st(v)(0) = \init{v}$ by the definition of $\st$, and
    $\init{v} \in \Anno{v}{0}$ since $\An$ satisfies the initial condition, so by substitution
    we have $\st(v)(0) \in \Anno{v}{0}$.

    For the inductive case, we must show that
    \begin{equation*}
      \big(\forall v \in V,~\st(v)(t) \in \Anno{v}{t}\big) \Rightarrow
      \big(\forall v \in V,~\st(v)(t+1) \in \Anno{v}{t+1}\big)
    \end{equation*}

    We assume the antecedent (the inductive hypothesis).
    Consider the in-neighbors $u_{1},\ldots,u_{n}$ of $v$.
    Since $\An$ satisfies the inductive condition~\eqref{eq:ind-vc},
    we have:
    \begin{align*}
      & \forall s_{1} \in \Anno{u_{1}}{t}, \forall s_{2} \in \Anno{u_{2}}{t}, \ldots, \forall s_{n} \in \Anno{u_{n}}{t},
      & \left(\init{v} \M \MM_{i \in \{1,\ldots,n\}} \T_{u_{i}v}(s_{i})\right) \in \Anno{v}{t+1}
    \end{align*}
    Then by our inductive hypothesis, we have that $\st(u_{i})(t) \in \Anno{u_{i}}{t}$ for all $u_{i}$,
    so we can instantiate the universal quantifiers with these routes
    and substitute, which gives us:
    \begin{equation*}
      \init{v} \M \MM_{i \in \{1,\ldots,n\}} \T_{u_{i}v}(\boxed{\st(u_{i})(t)}) \in \Anno{v}{t+1}
    \end{equation*}
    The left-hand side is equal to the
    definition of $\st(v)(t+1)$ in~\eqref{eq:next-st},
    so we have $\boxed{\st(v)(t+1)} \in \Anno{v}{t+1}$.
    Then $\forall t \in \nat, \forall v \in V, \st(v)(t) \in \Anno{v}{t}$.
    \end{proof}

  \begin{corollary}[Safety]\label{thm:sound-cor-apx}
    Let $\An$ satisfy initial and inductive conditions for all nodes.
    Let P satisfy the safety condition with respect to A for all nodes.
    Then $\forall v \in V, \forall t \in \nat, \st(v)(t) \in \PPr{v}{t}$.
  \end{corollary}
  \begin{proof}
    Let $v$ be a node and $t$ be a time.
    By Theorem~\ref{thm:sound}, $\st(v)(t) \in \Anno{v}{t}$.
    By the safety condition~\eqref{eq:safety-vc}, $\Anno{v}{t} \subseteq \PPr{v}{t}$.
    Then $\st(v)(t) \in \PPr{v}{t}$, \IE $\PP$ holds for $v$ at $t$.
  \end{proof}

  \begin{theorem}[Closed Simulation Completeness]\label{thm:complete-apx}
    Let $\st$ be the state of the \new{closed} network.
    Then for all $v \in V$ and $t \in \nat$, $\Anno{v}{t} = \{\st(v)(t)\}$
    satisfies the initial and inductive conditions \new{for all nodes}.
  \end{theorem}
  \begin{proof}
    Let $\An$ be a function 
      such that
    $\Anno{v}{t} = \{\st(v)(t)\}$ for all $v \in V$ and all $t \in \nat$.
    Let $v$ be an arbitrary node and $t$ be an arbitrary time.
    We want to show that $\{\st(v)(t)\}$ is an inductive invariant.

    Starting with the initial condition case,
    at time 0, we have $\st(v)(0) = \init{v}$ by the definition of $\st$.
    Since $\st(v)(0) \in \{\st(v)(0)\}$,
    the initial condition holds.

    For the inductive condition case,
    we want to show that:
      \begin{align*}
      & \forall s_{1} \in \boxed{\{\st(u_{1})(t)\}}, \forall s_{2} \in \boxed{\{\st(u_{2})(t)\}}, \ldots, \forall s_{n} \in \boxed{\{\st(u_{n})(t)\}}, \\
      & \left(\init{v} \M \MM_{i \in \{1,\ldots,n\}} \T_{u_{i}v}(s_{i})\right)
      \in \boxed{\{\st(v)(t+1)\}}
      \end{align*}
    We can enumerate the singleton sets (\IE substitute each $\st(u_{i})(t)$ for $s_{i}$) to give:
    \begin{align*}
		 & \left(\init{v} \M \MM_{i \in \{1,\ldots,n\}} \T_{u_{i}v}(\boxed{\st(u_{i})(t)})\right)
      \boxed{= \st(v)(t+1)}
    \end{align*}
    This is now simply~\eqref{eq:next-st} flipped,
    so the inductive condition holds.

    Then $\An$ satisfies the initial and inductive conditions for all nodes.
  \end{proof}
}{}

\end{document}